\documentclass[10pt]{article}

\usepackage{amssymb,mathtools,enumitem,mathrsfs} 
\usepackage{overpic,subcaption,hyperref}
\usepackage[linesnumbered,ruled,vlined]{algorithm2e}
\usepackage{todonotes}\reversemarginpar

\usepackage{geometry}

\newtheorem{theorem}{Theorem}

\newtheorem{lemma}[theorem]{Lemma}
\newtheorem{corollary}[theorem]{Corollary}
\newtheorem{remark}{Remark}
\newtheorem{definition}{Definition}
\newenvironment{proof}{\noindent{\bf Proof.}}
{\hspace*{\fill}$\Box$\par\vspace{4mm}}

\newcommand{\esterna}{f^{\infty}}
\newcommand{\uuu}{}
\newcommand{\uu}{}
\newcommand{\dd}{\overrightarrow}

\newcommand{\head} {\text{head}}

\newcommand{\Right} {\text{Right}}
\newcommand{\Left} {\text{Left}}

\newcommand{\PIPPOMaiuscolo} {Algorithm \texttt{ImplicitPaths}\xspace}
\newcommand{\PIPPO} {algorithm \texttt{ImplicitPaths}\xspace}

\begin{document}

\author{Lorenzo Balzotti\footnote{Dipartimento di Scienze di Base e Applicate per l’Ingegneria, Sapienza Universit\`a di Roma, Via Antonio Scarpa, 16, 00161 Roma, Italy. \texttt{lorenzo.balzotti@sbai.uniroma1.it}.}
\and
{Paolo G. Franciosa\footnote{Dipartimento di Scienze Statistiche, Sapienza Universit\`a di Roma, p.le Aldo Moro 5, 00185 Roma, Italy. \texttt{paolo.franciosa@uniroma1.it}}}}

\title{Computing Lengths of Non-Crossing Shortest Paths in Planar Graphs}

\date{}	
\maketitle

\begin{abstract} 
Given a plane undirected graph $G$ with non-negative edge weights and a set of $k$ terminal pairs
on the external face, it is shown in Takahashi \emph{et al.} (Algorithmica, 16, 1996, pp. 339–357) that
the union $U$ of $k$ non-crossing shortest paths joining the $k$ terminal pairs (if they exist) can be computed in $O(n\log n)$ time, where $n$ is the number of vertices of $G$. In the restricted case in which the union $U$ of the shortest paths is a forest, it is also shown that their lengths can be computed in the same time bound. We show in this paper that it is always possible to compute the lengths of $k$ non-crossing shortest paths joining the $k$ terminal pairs in linear time, once the shortest paths union $U$ has been computed, also in the case $U$ contains cycles.

Moreover, each shortest path $\pi$ can be listed in $O(\max\{\ell, \ell \log\frac{k}{\ell} \})$,
where $\ell$ is the number of edges in $\pi$.

As a consequence, the problem of computing non-crossing shortest paths and their lengths in a plane undirected weighted graph can be solved in $O(n\log k)$ time in the general case.
\end{abstract}

\textbf{keywords}: shortest paths, planar undirected graphs, non-crossing paths

\section{Introduction}
\label{section:introduction}
The problem of finding non-crossing shortest paths in a plane graph (i.e., a planar graph with a fixed embedding) has its primary applications in VLSI layout~\cite{bhatt-leighton}, where two paths are \emph{non-crossing} if they do not cross each other in the chosen embedding. It also appears as a basic step in the computation of maximum flow in a planar network and related problems~\cite{ausiello-balzotti,ausiello-franciosa_1,balzotti-franciosa_3,reif}. The problem can be formalized as follows: given an undirected plane graph $G$ with non-negative edge lengths and $k$ terminal pairs that lie
on a specified face boundary, find $k$ non-crossing shortest paths in $G$, each connecting a terminal pair. It is assumed that terminal pairs appear in the external face so that non-crossing paths exist, this property can be easily verified in linear time.

Takahashi \emph{et al.}~\cite{giappo2} proposed an algorithm able to compute $k$ non-crossing shortest paths that requires $O(n\log n)$ time, where $n$ is the size of $G$. In the same paper it is also analyzed the case where the terminal pairs lie on two different face boundaries, and this case is reduced to the previous one within the same computational complexity. The complexity of their solution can be reduced to $O(n\log k)$ by plugging in the linear time algorithm by Henzinger \emph{et al.}~\cite{henzinger} for computing a shortest path tree in a planar graph. 
In the unweighted case, by using the result of Eisenstat and Klein~\cite{eisenstat-klein}, Balzotti and Franciosa shown that $k$ non-crossing shortest paths can be found in $O(n)$ time~\cite{balzotti-franciosa_2}.

The algorithm proposed in~\cite{giappo2} first computes the union of the $k$ shortest paths, which is claimed to be a forest. The second step relies on a data structure due to Gabow and Tarjan~\cite{gabow-tarjan} for efficiently solving least common ancestor (LCA) queries in a forest, in order to obtain distances  between the terminal pairs in $O(n)$ time.

Actually, the union of the $k$ shortest paths may in general contain cycles. An instance is shown in Figure~\ref{fig:errore_giappo}, in which the unique set of $k$ shortest paths  contains a cycle, hence the distances between terminal pairs cannot always be computed by solving LCA queries in a forest. This limitation was noted first in~\cite{polishchuk-mitchell}.

\begin{figure}[h]
\centering
\begin{overpic}[width=4cm,percent]{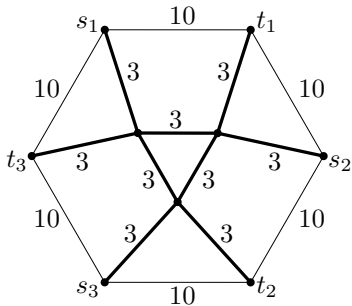}

\put(16,86){$s_1$}
\put(76,86){$t_1$}

\put(100,40){$s_2$}
\put(76,-2){$t_2$}

\put(16,-2){$s_3$}
\put(-7,40){$t_3$}

\put(33,69){$3$}
\put(63,69){$3$}
\put(80,39){$3$}
\put(64,15){$3$}
\put(32,15){$3$}
\put(16,39){$3$}

\put(47,53){$3$}
\put(58,33){$3$}
\put(38,33){$3$}

\put(47,87.5){$10$}
\put(90,63){$10$}
\put(90,20){$10$}
\put(47,-5.5){$10$}
\put(2,20){$10$}
\put(2,63){$10$}

\end{overpic}
\caption{in this example the union of shortest paths from $s_i$ to $t_i$, for $i=1,2,3$, contains a cycle (the union is highlighted with bold edges).} 
\label{fig:errore_giappo}
\end{figure}

In this paper we solve the problem in the general case, showing an algorithm that given a plane graph $U$ with non-negative edge weights and $k$ terminal pairs lying on the external face computes the $k$ distances between terminal pairs in time proportional to the dimension of $U$, provided that $U$ is the union of the shortest paths. In this way, our novel simple technique does not require the result of Gabow and Tarjan~\cite{gabow-tarjan}. Moreover, we also propose an algorithm for listing the sequence of edges in a shortest path $\pi$ joining a terminal pair in $O(\max\{\ell,\ell\log(\frac{k}{\ell})\})$, where $\ell$ is the number of edges in $\pi$. The listing algorithm can be applied to any terminal pair.

Our algorithm can replace the second step in the solution proposed in~\cite{giappo2}, obtaining an  $O(n \log k)$ time algorithm that solves any possible instance of the non-crossing shortest paths problem in a plane graph. 

The algorithm we propose can be easily implemented and it does not require sophisticated data structures. We follow the same approach of Polishchuck and Mitchell~\cite{polishchuk-mitchell}, that was inspired by Papadopoulou's work~\cite{papadopoulou}. Their paper solves the problem of finding $k$ non-crossing shortest paths in a polygon with $n$ vertices, where distances are defined according to the Euclidean metric. As in~\cite{polishchuk-mitchell,giappo2}, we first build the \emph{genealogy tree} $T_g$, that describes the order in which terminal pairs appear on the external face. Then we find, one at a time,  the shortest paths joining terminal pairs, according to a postorder visit of $T_g$.

\paragraph*{Our approach}
The main novelty in our solution is the definition of \emph{shortcuts}, that are portions of the boundary of a face that allow us to modify a path without increasing (and possibly decreasing) its length.  We show that it is possible to establish whether a path is a shortest path by looking at the presence of shortcuts. This is the main theoretical result of this paper: while being a shortest path is a global property,  we can check it locally by looking at a single face at a time for the presence of shortcuts adjacent to the path, ignoring the rest of the graph. Notice that this is only possible when the input graph is the union of non-crossing shortest paths, not for general plane graphs.

Using this result, we can introduce the \PIPPO that computes an implicit representation of non-crossing shortest paths. 
 This implicit representation is necessary to compute all distances between terminal pairs in linear time and to solve the problem of listing the edges of a single shortest path. This problem has already been discussed in the geometrical case~\cite{papadopoulou,polishchuk-mitchell} with Euclidean distances, but it is new in non-crossing paths in a plane weighted graph which have a more general structure.
 
Shortest paths computed by \PIPPO fulfill the \emph{single-touch} property, i.e., the intersection of any pair of paths is itself a path. The single-touch property implies that non-crossing shortest single-touch paths are non-crossing for every planar embedding of $U$, provided that a non-crossing solution exists. For this reason, we can say that the solution found by our algorithm holds for any feasible planar embedding of the graph.

\paragraph*{Structure}
The paper is organized as follows: in Section~\ref{section:preliminaries} we give preliminary definitions and notations that will be used in the whole paper. In Section~\ref{section:shortcuts} we deal with shortcuts and in Section~\ref{section:PIPPO_and_correctness} we introduce \PIPPO that gives an implicit representation of the shortest paths. In Section~\ref{section:computational_complexity_PIPPO} we deal with its computational complexity and in Section~\ref{section:compute_distances_and_listing_paths} we use its output to compute distances and list non-crossing paths. Finally, in Section~\ref{section:conclusions} conclusions are given and we mention some open problems.

\section{Preliminaries}
\label{section:preliminaries}

General definitions and notations are given. Then we deal with paths and non-crossing paths. We define a partial order on terminal pairs, the \emph{genealogy tree}, and we introduce \emph{path-sets}, that are sets of non-crossing shortest paths. 
All graphs in this paper are undirected.

\subsection{Definitions and notations}
\label{section:definitions_and_notations}

Let $G=(V(G),E(G))$ be a plane graph, i.e., a planar graph with a fixed planar embedding. We denote by $\mathcal{F}_G$ the set of faces and by $f^\infty_G$ its external face. When no confusion arises we use the term face to denote both the border cycle and the finite region bounded by the border cycle, and the external face is simply denoted by $f^\infty$.
 

We use standard union and intersection operators on graphs. 

\begin{definition}
Given two graphs $G=(V(G),E(G))$ and $H=(V(H),$ $E(H))$, we define the following operations and relations:
\begin{itemize}\itemsep0em
\item $G\cup H=(V(G)\cup V(H),E(G)\cup E(H))$,
\item $G\cap H=(V(G)\cap V(H),E(G)\cap E(H))$,
\item $H\subseteq G\Longleftrightarrow V(H)\subseteq V(G)$ $\wedge$ $E(H)\subseteq E(G)$,
\item $G\setminus H=(V(G),E(G)\setminus E(H))$.
\end{itemize}
\end{definition}

Given a graph $G=(V(G),E(G))$, an edge $e$ and a vertex $v$ we write, for short, $e\in G$ in place of $e\in E(G)$ and $v\in G$ in place of $v\in V(G)$.

We denote by $uv$ the edge whose endpoints are $u$ and $v$.
%
%
For every vertex $v\in V(G)$ we define the \emph{degree of $v$} as $deg(v)=|\{e\in E(G) \,|\, \text{$v$ is an endpoint}$ $\text{of $e$}\}|$.

A graph $p$ is a \emph{path} from $a$ to $b$ if the set of edges in $p$ forms a sequence $av_1,v_1v_2,\ldots,v_{r-1}v_r,$ $v_rb$; we also say $p$ is an \emph{$ab$ path}. 
A path is \emph{simple} if each vertex has degree at most two. A \emph{cycle} is a path from $a$ to $a$, and it is a \emph{simple cycle} if all vertices have degree two.

We use round brackets to denote ordered sets. For example, $\{a,b,c\}=\{c,a,b\}$ and $(a,b,c)\neq (c,a,b)$. Moreover, for every $k\in\mathbb{N}$ we denote by $[k]$ the set $\{1,\ldots,k\}$.

Let $w:E(G)\rightarrow \mathbb{R}_{>0}$ be a weight function on edges. The weight function is extended to a subgraph $H$ of $G$ so that $w(H)=\sum_{e\in E(H)} w(e)$. W.l.o.g., we assume that $w$ is strictly positive, indeed if an edge has zero weight, then we can delete it and join its extremal vertices. In case a zero weight cycle $C$ exists, then all edges internal to $C$ are not relevant for determining shortest paths, and can be contracted in a single vertex.

We assume that the input of our problem is a plane undirected graph $U=\bigcup_{i\in[k]}\uu{p_i}$, where $p_i$ is a shortest $s_it_i$ path in a planar graph $G$, and the terminal pairs $\{(s_i,t_i)\}_{i\in[k]}$ lie on the external face $f^\infty$ of $U$. We stress that we work with a fixed embedding of $U$. W.l.o.g. we assume that $U$ is connected, otherwise it suffices to work on each connected component.  All our algorithms work with $U$, while the original graph $G$ is not required.

For a (possibly not simple) cycle $C$, we define the \emph{region bounded by $C$} the maximal subgraph of $U$ whose external face has $C$ as boundary.

W.l.o.g., we assume that the terminal pairs are distinct, i.e., there is no pair $i,j\in[k]$ such that $\{s_i,t_i\}=\{s_j,t_j\}$. Let $\gamma_i$ be the path in $f^\infty$ that goes clockwise from $s_i$ to $t_i$, for $i\in[k]$. By hypothesis, we assume that pairs $\{(s_i,t_i)\}_{i\in[k]}$ are \emph{well-formed}, i.e., for all $j,\ell\in[k]$ either ${\gamma_j}\subset{\gamma_\ell}$ or ${\gamma_j}\supset{\gamma_\ell}$ or ${\gamma_j}$ and ${\gamma_\ell}$ share no edges. We note that there exists a set of pairwise non-crossing $s_it_i$ paths if and only if terminal pairs are well-formed.

We suppose that terminal vertices are pendant in $f^\infty$. More clearly, for every $x\in\{s_i,t_i\}_{i\in[k]}$, we assume that there exists an edge $xx'$ for a vertex $x'\in f^\infty$. Despite this, for drawing simplicity, we draw $f^\infty$ as a simple cycle in all figures. It is immediate to see that any path from $x$ must contain edge $xx'$, hence the relation among lengths of paths is the same as if terminal vertices were contained in the external face (contracting edge $xx'$).

Given $i\in[k]$, we denote by \emph{$i$-path} an $s_it_i$ path. It is always useful to see each $i$-path as oriented from $s_i$ to $t_i$, for $i\in[k]$, even if the path is undirected. For an $i$-path $p$, we define $\Left_p$ as the left portion of $U$ with respect to $p$, i.e., the finite region bounded by the cycle formed by $\uu p$ and $\gamma_i$; similarly, we define $\Right_p$ as the right portion of $U$ with respect to $p$, i.e., the finite region bounded by the cycle formed by $\uu p$ and $f^\infty \setminus \gamma_i$.

For an $i$-path $p$ and a $j$-path $q$, we say that $q$ \emph{is to the right of} $p$ if $q\subseteq\Right_p$, similarly, we say that $q$ \emph{is to the left of} $p$ if $q\subseteq\Left_p$.

Given $R\subseteq U$ and an $i$-path $p\subseteq R$, for some $i\in[k]$, we say that $p$ is the \emph{leftmost $i$-path in $R$} if $p$ is to the left left of $q$ for each $i$-path $q\subseteq R$. Similarly, we say that $p$ is the \emph{rightmost $i$-path in $R$} if $p$ is to the right of $q$ for each $i$-path $q\subseteq R$.

If $R$ is a subgraph of $U$, then we denote by $\partial R$ the external face of $R$. Moreover, we define $\mathring{R}=R\setminus\partial R$. 

\subsection{Paths and non-crossing paths}
\label{section:paths_and_non-crossing_paths}

Given an $ab$ path $p$ and a $bc$ path $q$, we define $p\circ q$ as the (possibly not simple) path obtained by the  union of $p$ and $q$. 

Given a simple path $p$ and two vertices $x,y$ of $p$, we denote by $p[x,y]$ the subpath of $p$ with extremal vertices $x$ and $y$.

Now we introduce an operator that allows us to 
replace a subpath in a path.

\begin{definition}
Let $p$ be a simple $ab$ path, let $u,v\in V(p)$ such that $a,u,v,b$ appear in this order in $p$ and let $q$ be a $uv$ path. We denote by $p\rtimes q$ the (possibly not simple) path $p[a,u]\circ q\circ p[v,b]$.
\end{definition}

Figure~\ref{fig:rtimes} shows two examples of operator $\rtimes$.

\begin{figure}[h]
\captionsetup[subfigure]{justification=centering}
\centering
	\begin{subfigure}{4cm}
\begin{overpic}[scale=1.1,percent]{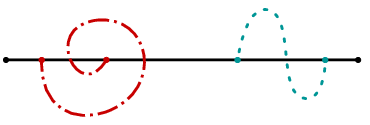}
\put(0,23.5){$a$}
\put(96,23.5){$b$}
\put(9,23.5){$u$}
\put(27.5,23.5){$v$}
\put(62,15){$y$}
\put(87,23.5){$x$}
\put(50,24){$p$}
\put(21,1){$q$}
\put(77.5,32){$r$}
\end{overpic}
\end{subfigure}
\qquad
	\begin{subfigure}{4cm}
\begin{overpic}[scale=1.1,percent]{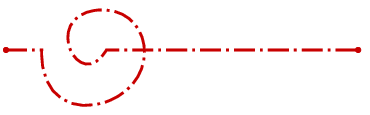}
\put(55,24){$p\rtimes q$}
\end{overpic}
\end{subfigure}
\qquad
	\begin{subfigure}{4cm}
\begin{overpic}[scale=1.1,percent]{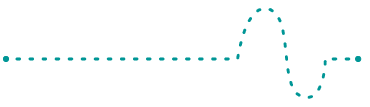}
\put(35,24){$p\rtimes r$}
\end{overpic}
\end{subfigure}	
  \caption{illustrating operator $\rtimes$.}
\label{fig:rtimes}
\end{figure}

The following definition, explained in Figure~\ref{fig:envelopes}, introduces right envelopes and left envelopes, that are important for transforming paths.

\begin{definition}
Let $p$ be an $i$-path and let $q$ be a $j$-path, for some $i,j\in[k]$. We define $p RE q$ (resp., $p LE q$) as the rightmost (resp., leftmost) $i$-path in $p\cup q$.
\end{definition}

\begin{figure}[h]
\captionsetup[subfigure]{justification=centering}
\centering
	\begin{subfigure}{3.5cm}
\begin{overpic}[width=3.5cm,percent]{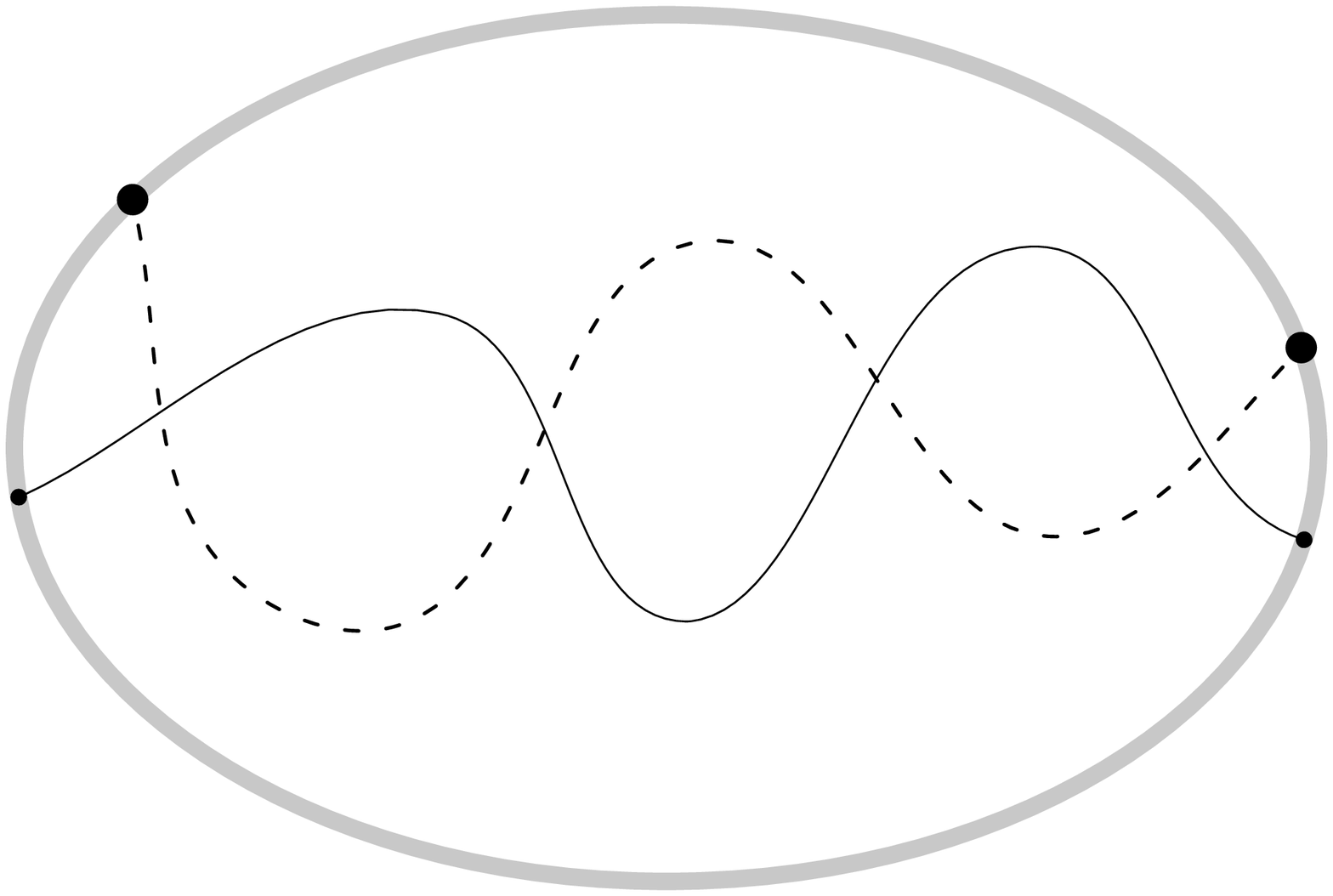}
\put(100,25){$s_i$}
\put(-7.5,30){$t_i$}
\put(0,53.7){$s_j$}
\put(100,40.5){$t_j$}
\put(69,50){$p$}
\put(25,15){$q$}
\end{overpic}
\end{subfigure}
\qquad\quad
	\begin{subfigure}{3.5cm}
\begin{overpic}[width=3.5cm,percent]{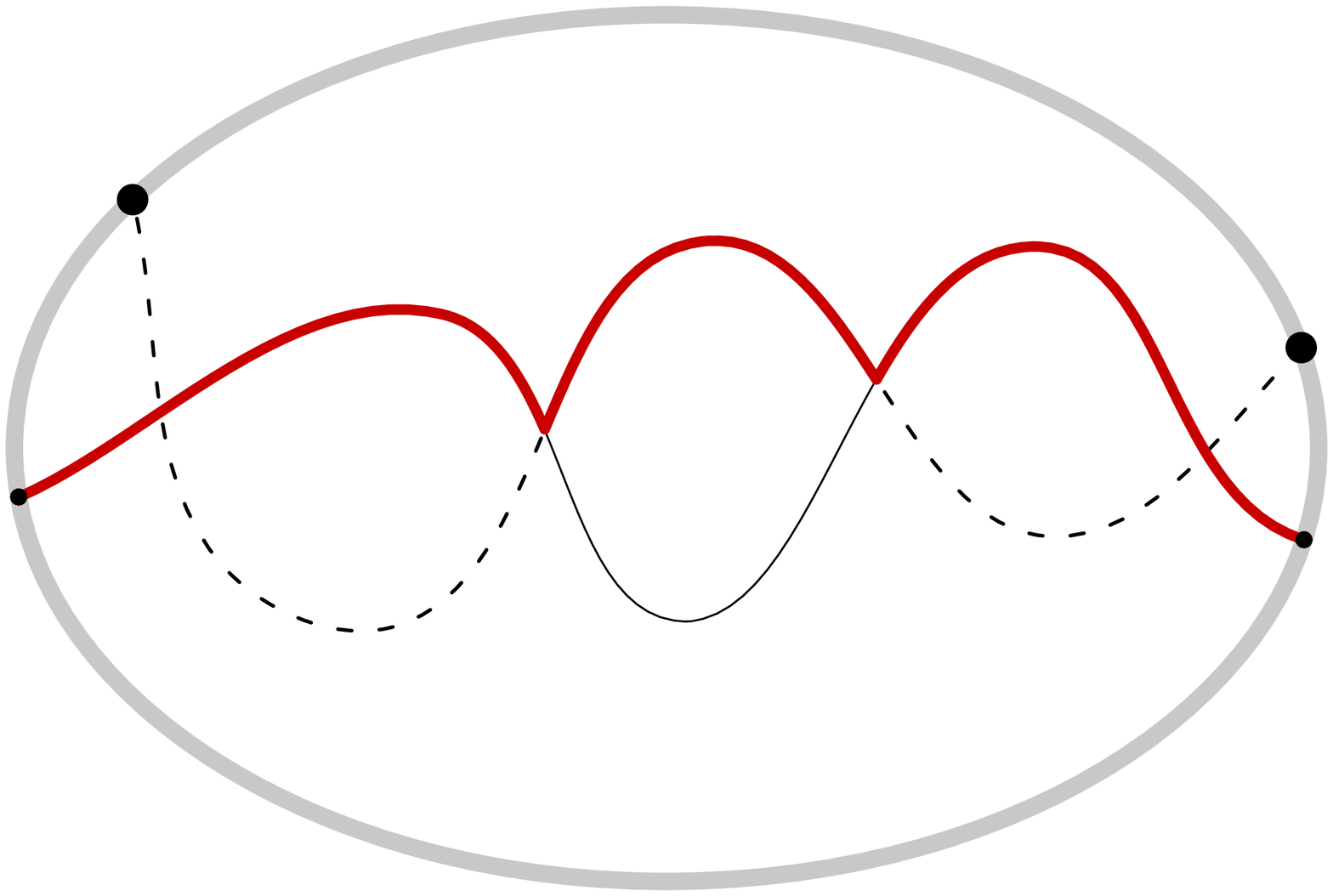}
\put(100,25){$s_i$}
\put(-7.5,30){$t_i$}
\put(0,53.7){$s_j$}
\put(100,40.5){$t_j$}
\put(40,51){$p RE q$}
\end{overpic}
\end{subfigure}
\qquad\quad
	\begin{subfigure}{3.5cm}
\begin{overpic}[width=3.5cm,percent]{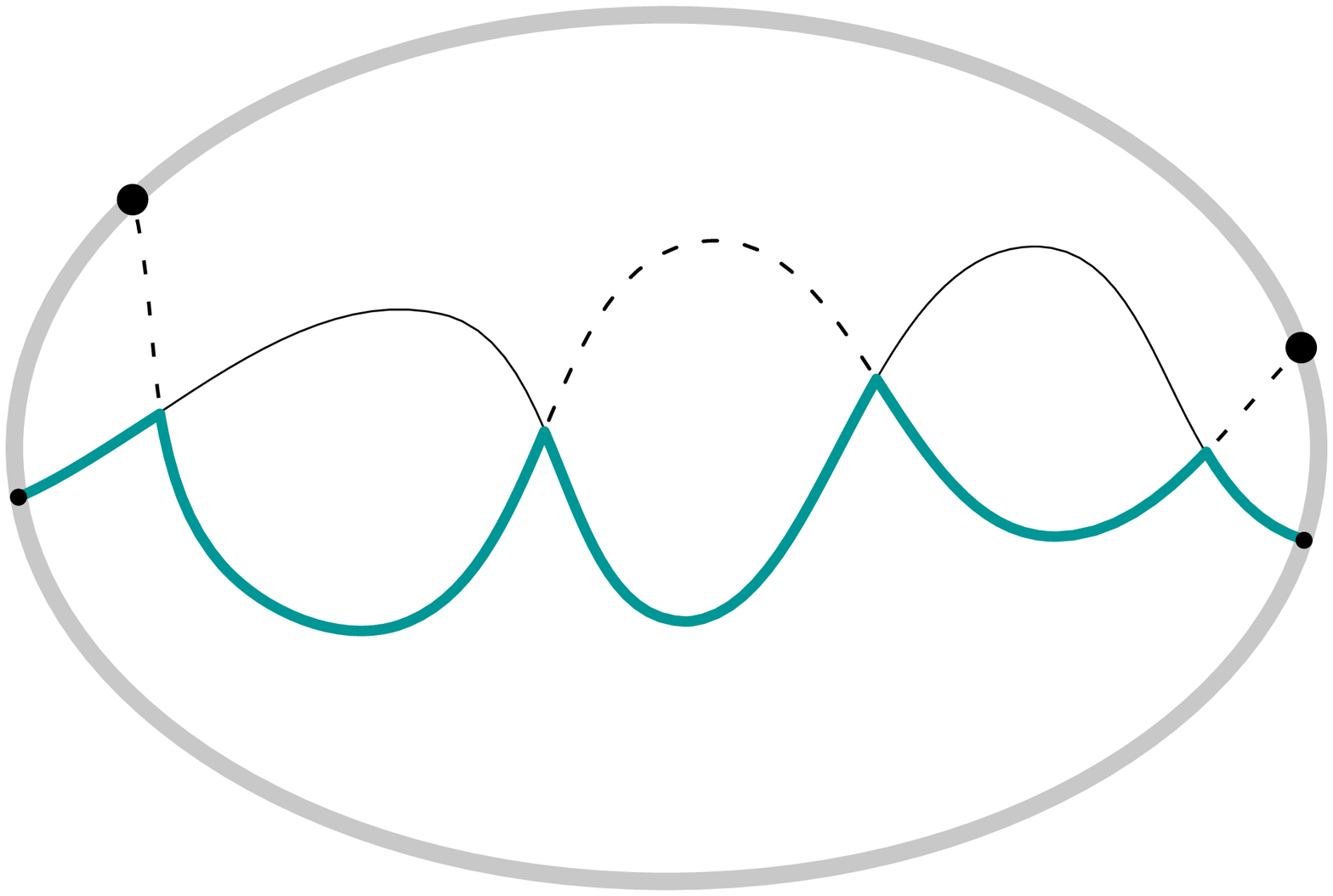}
\put(100,25){$s_i$}
\put(-7.5,30){$t_i$}
\put(0,53.7){$s_j$}
\put(100,40.5){$t_j$}
\put(45,10){$p LE q$}
\end{overpic}
\end{subfigure}	
  \caption{an example of right envelope and left envelope.}
\label{fig:envelopes}
\end{figure}

We say that two paths in a plane graph $G$ are \emph{non-crossing} if the curves they describe in the graph embedding do not cross each other; a combinatorial definition of
non-crossing paths can be based on the \emph{Heffter-Edmonds-Ringel rotation principle}~\cite{pisanski-potocnik}.  We stress that this property depends on the embedding of the graph. Non-crossing paths may share vertices and/or edges. We also define a class of paths that will be used later. 

\begin{definition}
Two paths $p$ and $q$ are \emph{single-touch} if $\uuu p\cap \uuu q$ is a (possibly empty) path.
\end{definition}

Examples of non-crossing paths and single-touch paths are given in Figure~\ref{fig:non_crossing_and_single-touch}.

\begin{figure}[h]
\captionsetup[subfigure]{justification=centering}
\centering
	\begin{subfigure}{2.6cm}
\begin{overpic}[width=2.6cm,percent]{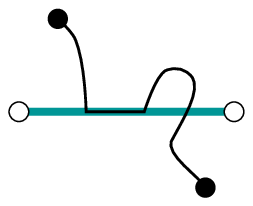}
\end{overpic}
\caption{}\label{fig:non-crossing_a}
\end{subfigure}
	\begin{subfigure}{2.6cm}
\begin{overpic}[width=2.6cm,percent]{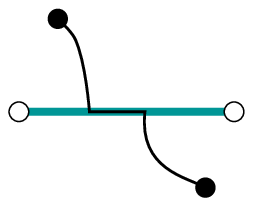}
\end{overpic}
\caption{}\label{fig:non-crossing_e}
\end{subfigure}
	\begin{subfigure}{2.6cm}
\begin{overpic}[width=2.6cm,percent]{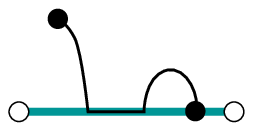}
\end{overpic}
\caption{}\label{fig:non-crossing_b}
\end{subfigure}	
	\begin{subfigure}{2.6cm}
\begin{overpic}[width=2.6cm,percent]{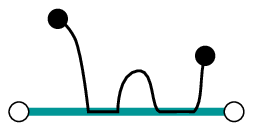}
\end{overpic}
\caption{}\label{fig:non-crossing_b_bis}
\end{subfigure}	
	\begin{subfigure}{2.6cm}
\begin{overpic}[width=2.6cm,percent]{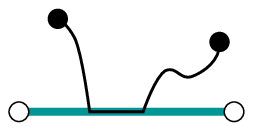}
\end{overpic}
\caption{}\label{fig:non-crossing_c}
\end{subfigure}	
  \caption{paths in (\subref{fig:non-crossing_a}) and (\subref{fig:non-crossing_e}) are crossing, while paths in (\subref{fig:non-crossing_b}), (\subref{fig:non-crossing_b_bis}), (\subref{fig:non-crossing_c}) are non-crossing. Moreover, paths in (\subref{fig:non-crossing_a}), (\subref{fig:non-crossing_b}) and (\subref{fig:non-crossing_b_bis}) are not single-touch, while paths in (\subref{fig:non-crossing_e}) and (\subref{fig:non-crossing_c}) are single-touch.}
\label{fig:non_crossing_and_single-touch}
\end{figure}

We note that the single-touch property does not depend on the embedding, and if the terminal-pair are well-formed, then it implies the non-crossing property. This is explained in the following trivial remark.

\begin{remark}
\label{remark:noncrossing_for_all_embeddings}
If $\{\pi_i\}_{i\in[k]}$ is a set of simple single-touch undirected paths, where $\pi_i$ is an $i$-path, for $i\in[k]$, then $\{\pi_i\}_{i\in[k]}$ is a set of pairwise non-crossing paths for all the embeddings of $U$ such that the terminal pairs $\{(s_i,t_i)\}_{i\in[k]}$ are well-formed.
\end{remark}

In the general case, $U$ might be the union of a set of shortest paths containing crossing pairs of paths and/or not single-touch pairs of paths. This may happen if there are  more shortest paths in the graph joining the same pair of vertices.  Uniqueness of shortest paths can be easily ensured by introducing small perturbations in the weight function. We wish to point out that the technique we describe in this paper does not rely on perturbation, but we break ties by choosing rightmost or leftmost paths. This implies that our results can be used also in the unweighted case, as already done by Balzotti and Franciosa in~\cite{balzotti-franciosa_2}.

\subsection{Genealogy tree}
\label{section:genealogy_tree}

Given a well-formed set of pairs $\{(s_i,t_i)\}_{i\in[k]}$, we define here a partial ordering as in~\cite{giappo2} that represents the inclusion relation between $\gamma_i$'s. This relation intuitively corresponds to an \emph{adjacency} relation between non-crossing shortest  paths joining each pair.

Choose an arbitrary $i^*$ such that there are neither $s_{j}$ nor $t_{j}$, with $j\neq i^*$, walking on $\esterna$ from $s_{i^*}$ to $t_{i^*}$ (either clockwise or counterclockwise), and let $e^*$ be an arbitrary edge on that walk. For each $j \in [k]$, we can assume that $e^*\not\in\gamma_j$, indeed if it is not true, then it suffices to switch $s_j$ with $t_j$. We say that $i \prec j$ if $\gamma_i\subset\gamma_j$. We define the \emph{genealogy tree}  $T_g$ of a set of well formed terminal pairs as the transitive reduction of poset $([k],\prec)$.

If $i\prec j$, then we say that $i$ is a \emph{descendant} of $j$ and $j$ is an \emph{ancestor} of $i$. Given $i,j\in[k]$, we say that $i$ and $j$ are \emph{uncomparable} if $i\not\prec j$ and $j\not\prec i$. Finally, we denote by $p(j)$ the \emph{parent of $j$}, i.e.,  the lowest ancestor of $j$ with respect to $\prec$.

Figure~\ref{fig:gt} shows an example of well-formed terminal pairs, and the corresponding genealogy tree for $i^*=1$.  From now on, in all figures we draw $f^\infty_G$ by a solid light grey line.

\begin{figure}[h]
\captionsetup[subfigure]{justification=centering}
\centering
	\begin{subfigure}{4cm}
\begin{overpic}[width=6cm,percent]{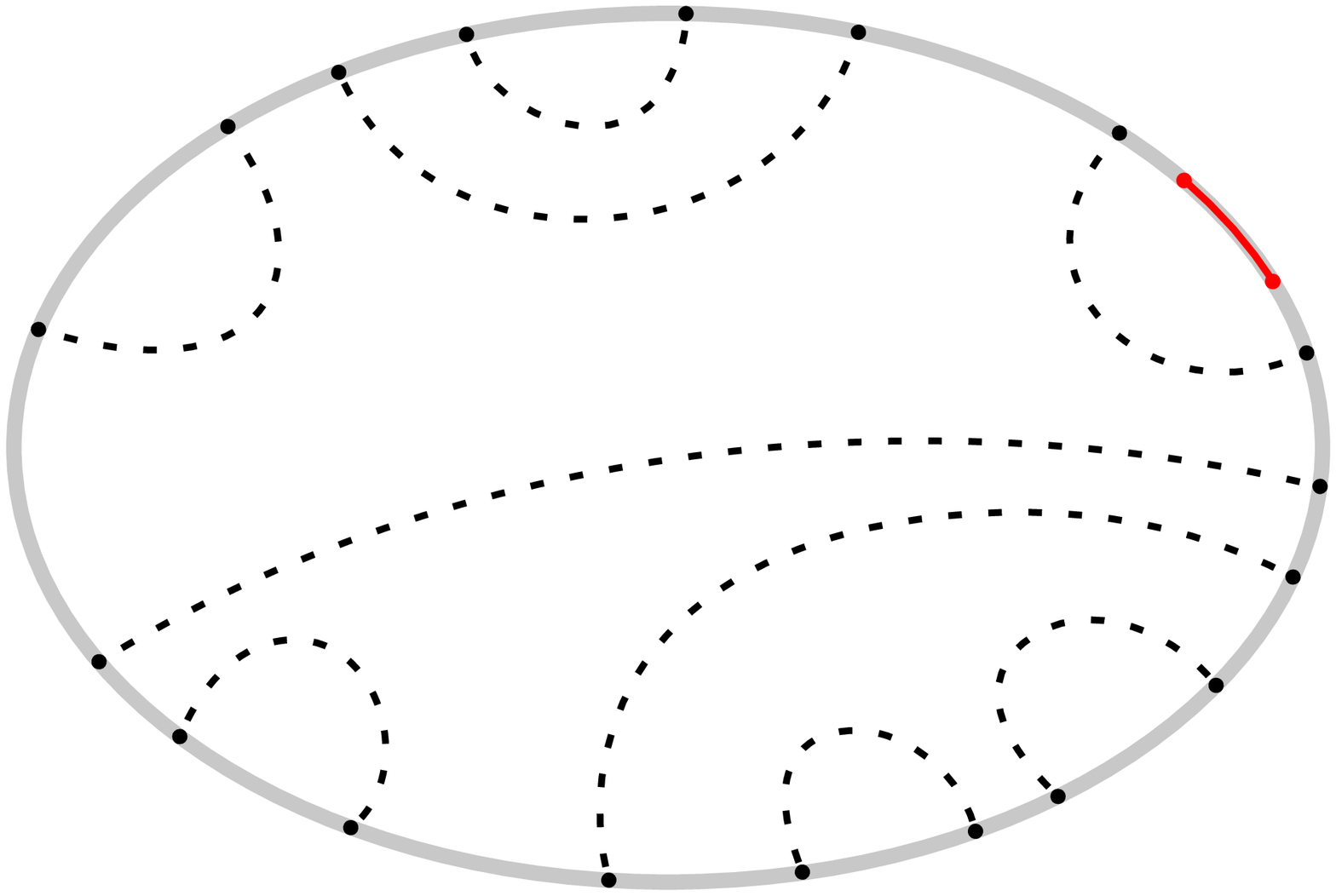}
\put(91,51){\color{red} $e^*$}
\put(-4.5,32){$f^\infty_G$}

\put(82,58){$t_1$}
\put(96,42){$s_1$}
\put(97.5,30.5){$s_2$}
\put(7,16.5){$t_2$}
\put(95.4,24){$s_3$}
\put(44.5,0){$t_3$}
\put(89.9,16.5){$s_4$}
\put(79,7){$t_4$}
\put(70.2,4.2){$s_5$}
\put(60,0.5){$t_5$}
\put(25,5){$s_6$}
\put(15,10){$t_6$}
\put(0,43.5){$s_7$}
\put(15.5,59){$t_7$}
\put(24,63.5){$s_8$}
\put(64,65.6){$t_8$}
\put(34,66.5){$s_9$}
\put(50,67.5){$t_9$}

\end{overpic}
\end{subfigure}
\qquad\qquad\qquad\qquad\quad
	\begin{subfigure}{3.6cm}
\begin{overpic}[width=3.6cm,percent]{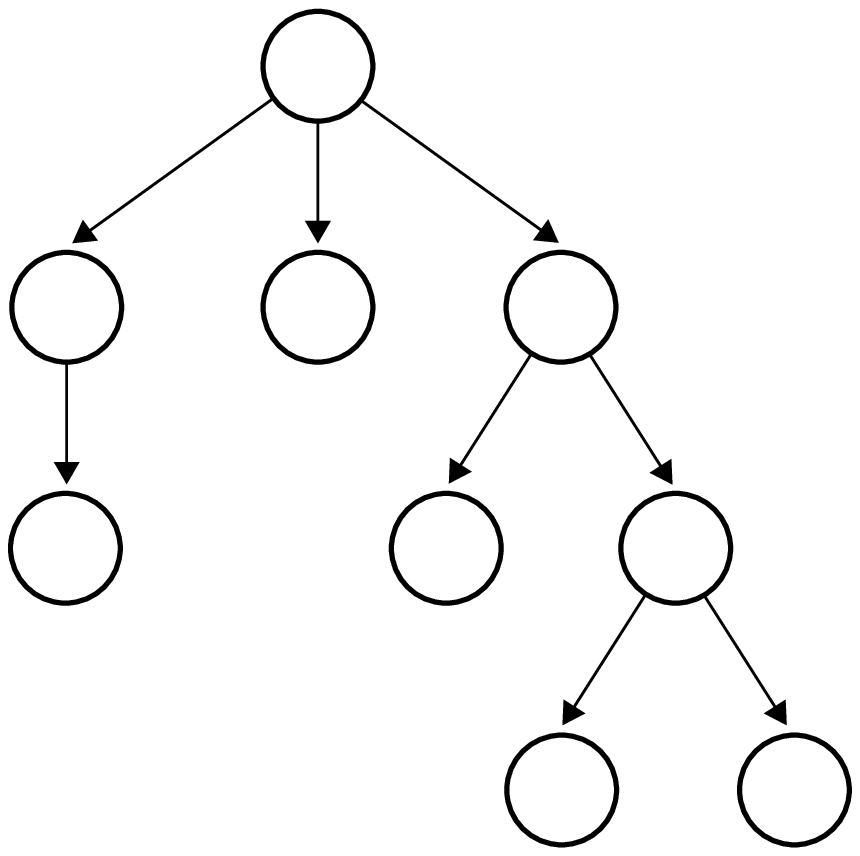}
\put(34.6,90){$1$}
\put(5,61){$8$}
\put(34.6,61){$7$}
\put(63.6,61){$2$}
\put(5,32.9){$9$}
\put(49.9,32.9){$6$}
\put(77.1,32.9){$3$}
\put(63.1,4.2){$5$}
\put(90.4,4.2){$4$}
\end{overpic}
\end{subfigure}
\caption{on the left an example of well-formed terminal pairs. If we choose $i^*=1$, then we obtain the genealogy tree on the right.}
\label{fig:gt}
\end{figure}

\subsection{Path-sets}
\label{section:path-sets}

We remind that $U$ is the union of shortest $i$-paths, for $i\in[k]$. That is, there exist $p_1,\ldots,p_k$ such that $p_i$ is a shortest $i$-path, for $i\in[k]$, and $U=\bigcup_{i\in[k]}\uu{p_i}$. Despite this, it is possible that a set of shortest paths $p'_1,\ldots,p'_k$ exists such that $p_i'$ is a shortest $i$-path, for $i\in[k]$, $U=\bigcup_{i\in[k]}\uu{p_i'}$  but  $p_j\neq p'_j$ for at least one $j\in[k]$.
In other words, the union of shortest paths does not uniquely correspond to a set of shortest paths. To overcome this ambiguity, we introduce the definition of path-set.

We say that a set of paths $P=\{p_1,\ldots,p_k\}$ is a \emph{path-set} if for all $i\in[k]$, $p_i$ is a shortest $i$-path and, for all $j,\ell\in[k]$, $p_j$ and $p_\ell$ are non-crossing. We denote the set of all possible path-sets by $\mathcal{P}$. If a path-set $P\in\mathcal{P}$ satisfies $\bigcup_{p\in P} p=U$, then we call $P$ a \emph{complete path-set}. We underline that $U$ might be obtained by the union of crossing paths. Despite this, by using right and left envelopes, it is easy to see that a complete path-set always exists.


\section{Shortcuts}
\label{section:shortcuts}
Now we have all the necessary machinery to introduce \emph{shortcuts}, that are the main tool of \PIPPO introduced in Section~\ref{section:PIPPO_and_correctness}, and the most important theoretical novelty of this paper.

Roughly speaking, a shortcut appears if there exists a face $f$ adjacent to a path $\pi$ so that we can modify $\pi$ going around $f$ without increasing its length. We show that we can decide whether a path is a shortest path  by looking at the existence of shortcuts: in this way, we can check a global property of a path $\pi$---i.e., being a shortest path---by checking a local property---i.e., the presence of shortcuts in faces adjacent to $\pi$.
This result is not true for general plane graphs, but it only holds when the input graph is the union of shortest paths joining well-formed terminal pairs.

\begin{definition}
\label{def:shortcut}
Given a path $\lambda$ and a face $f$ containing two vertices $u,v\in\lambda$, we say that a $uv$ subpath $q$ of $f$ not contained in $\lambda$ is a \emph{shortcut of} $\lambda$ if $w(\lambda\rtimes q)\leq w(\lambda)$.
\end{definition}

Figure~\ref{fig:shortcut} explains Definition~\ref{def:shortcut}.
\begin{figure}[h]
\captionsetup[subfigure]{justification=centering}
\centering
	\begin{subfigure}{3.5cm}
\begin{overpic}[width=3.5cm,percent]{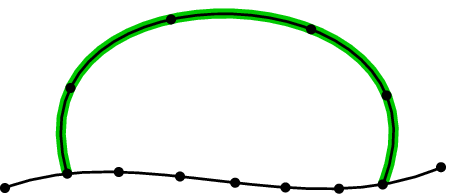}
\put(54,-6.5){$\lambda$}
\put(50,20){$f$}
\end{overpic}
\end{subfigure}
\qquad
	\begin{subfigure}{3.5cm}
\begin{overpic}[width=3.5cm,percent]{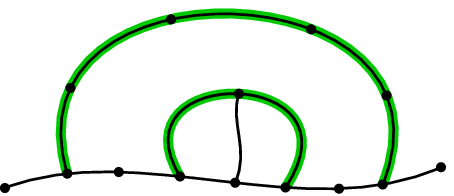}
\put(54,-6.5){$\lambda$}
\put(50,30){$f$}
\put(44,10){$g$}
\put(57,10){$h$}

\end{overpic}
\end{subfigure}
\qquad
	\begin{subfigure}{3.5cm}
\begin{overpic}[width=3.5cm,percent]{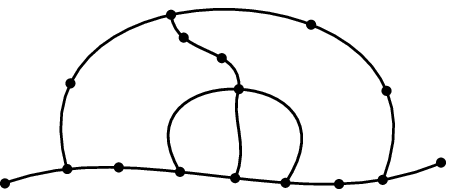}
\put(54,-6.5){$\lambda$}
\put(25,20){$f$}
\put(44,10){$g$}
\put(57,10){$h$}
\put(70,23){$k$}
\end{overpic}
\end{subfigure}	
  \caption{all edges have unit weight. On the left, highlighted in green, there is a shortcut  for $\lambda$ contained in $\partial f$. In the middle there are two shortcuts for $\lambda$ both contained in $\partial f$. On the right there are no shortcut for $\lambda$.}
\label{fig:shortcut}
\end{figure}



\begin{theorem}\label{th_shortcut}
Let $\lambda$ be an $i$-path, for some $i\in[k]$. If there are no shortcut of $\lambda$, then $\lambda$ is a shortest $i$-path.
\end{theorem}
\begin{proof}
Let $P\in\mathcal{P}$ be a complete path-set and let $\tau$ be the unique shortest $i$-path in $P$. If $\tau=\lambda$ then the thesis holds. Thus let us assume by contradiction that $w(\tau)<w(\lambda)$ and $\lambda$ has no shortcuts. 

Let $a,b\in V(\lambda)\cap V(\tau)$ be two vertices such that $\tau[a,b]$ and $\lambda[b,a]$ share no edges (such $a$ and $b$ exist because $\tau\neq\lambda$ and they are both $i$-path). Let $C$ be the simple cycle $\tau[a,b]\circ \lambda[b,a]$, and let $R$ be the region bounded by $C$. If $R$ is a face of $U$, then $\tau[a,b]$ is a shortcut for $\lambda$, absurdum. Hence we assume that there exist edges in $\mathring{R}$, see Figure~\ref{fig:tau} on the left.

Either $R\subseteq\Left_\tau$ or $R\subseteq\Right_\tau$. W.l.o.g., we assume that $R\subseteq\Left_\tau$. Being $P$ a complete path-set, for every edge $e\in\mathring{R}$ there exists at least one path $q\in P$ such that $e\in q$. Moreover, the extremal vertices of $q$ are in $\gamma_i$ because paths in $P$ are non-crossing and $R\subseteq\Left_\tau$. Indeed, we recall that the extremal vertices are pending vertices (even if in the figures, for convenience, we draw the infinite face of $G$ as a simple cycle).

Now we show by construction that there exist a path $p \in P$ and a face $f$ such that $f\subseteq R$, $\partial f$ intersects $\lambda$ on vertices and $\partial f\setminus\lambda\subseteq p$; thus $\partial f\setminus\lambda$ is a shortcut of $\lambda$ because $p$ is a shortest path.

For all $q\in P$ such that $q\subseteq\Left_\tau$ we assume that $\Left_q\subseteq\Left_\tau$ (if it is not true, then it suffices to switch the extremal vertices of $q$).

For every $q\subseteq\Left_\tau$, let $F_q=\{f\in\mathcal{F}_U \,|\, \partial f\subseteq R\cap\Left_q\}$. To complete the proof, we have to find a path $p$ such that $|F_p|=1$, indeed the unique face $f$ in $F_p$ satisfies $\partial f\setminus\lambda\subseteq p$, and thus $\partial f\setminus\lambda$ is a shortcut for $\lambda$.

Now, let $e_1\in\mathring R$ and let $q_1\in P$ be such that $e_1\in q_1$. Being $e_1\in\mathring R$, then $|F_{q_1}|<|F_\tau|$ and $|F_{q_1}|>0$ because $e_1\in q_1$, see Figure~\ref{fig:tau} on the right. If $|F_{q_1}|=1$, then we have finished, otherwise we choose $e_2\in \mathring{R}\cap\mathring{\Left}_{q_1}$ and $q_2\in P$ such that $e_2\in q_2$. It holds that $|F_{q_2}|<|F_{q_1}|$ and $|F_{q_2}|>0$ because $e_2\in q_2$. By repeating this reasoning, and being $P$ a complete path-set, we find a path $p$ such that $|F_p|=1$.
\end{proof}

\begin{figure}[h]
\captionsetup[subfigure]{justification=centering}
\centering
	\begin{subfigure}{5cm}
\begin{overpic}[width=5cm,percent]{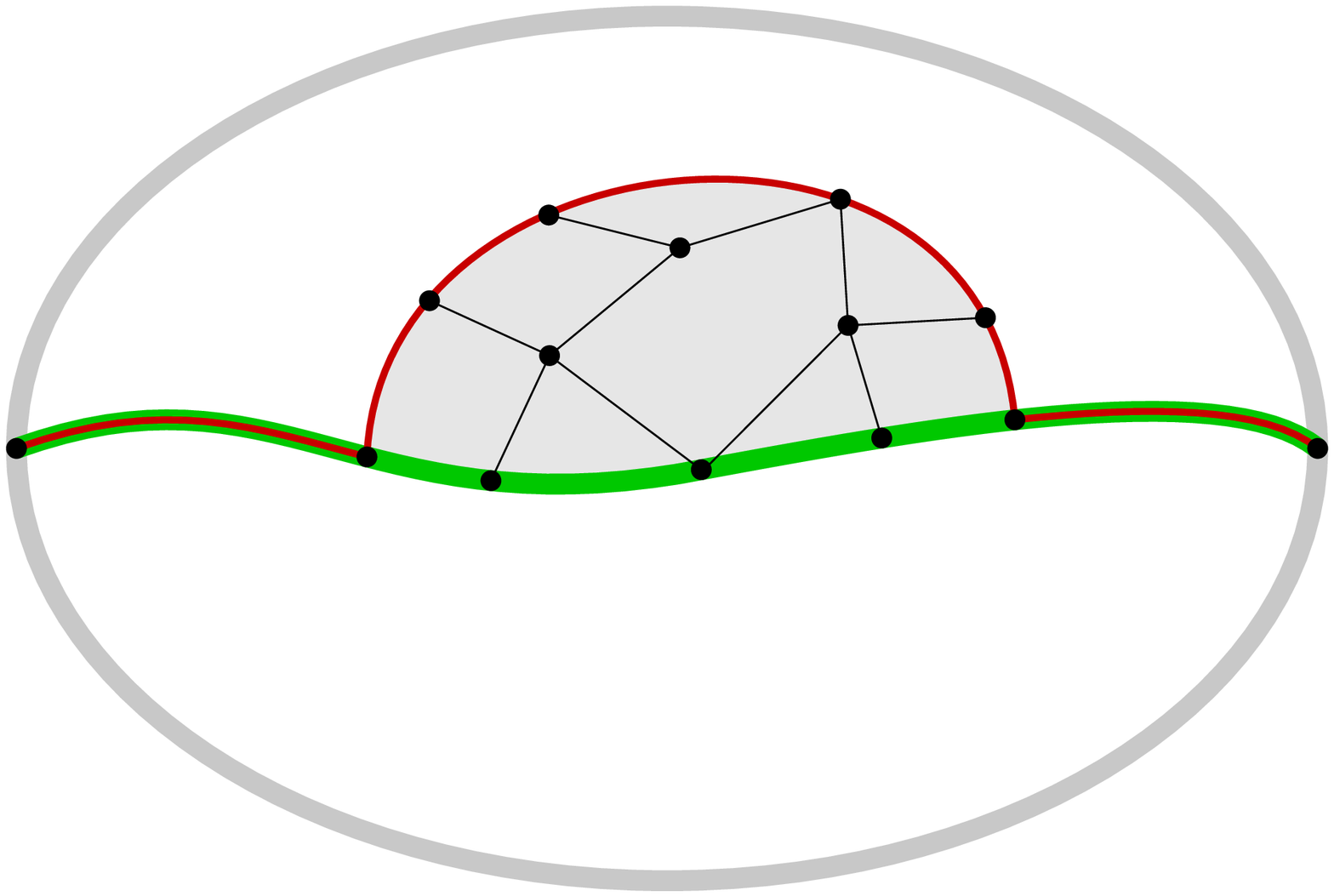}
\put(101,32){$s_i$}
\put(-6,32){$t_i$}
\put(75,30){$a$}
\put(26,26){$b$}

\put(55,55){$\tau$}
\put(57,26){$\lambda$}
\put(50,40){$R$}
\end{overpic}
\end{subfigure}
\qquad\qquad
	\begin{subfigure}{5cm}
\begin{overpic}[width=5cm,percent]{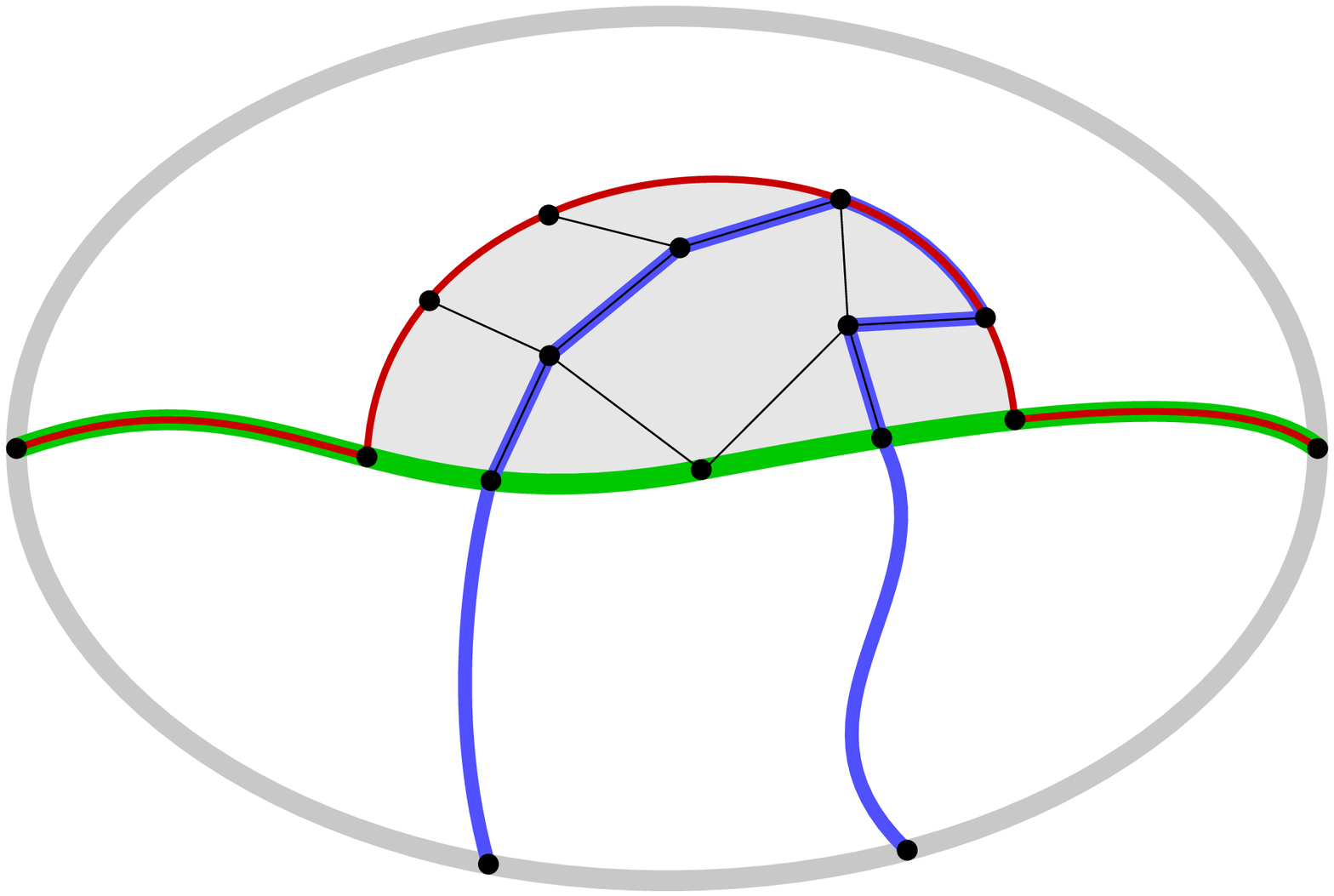}
\put(101,32){$s_i$}
\put(-6,32){$t_i$}
\put(55,55){$\tau$}
\put(57,26){$\lambda$}
\put(50,40){$R$}
\put(27.8,16){$q_1$}

\end{overpic}
\end{subfigure}
  \caption{paths and regions used in Theorem~\ref{th_shortcut}'s proof. Path $\lambda$ is in green, $\tau$ in red, $q_1$ in blue and region $R$ is highlighted in grey. It holds that $|F_{q_1}|=4$.}
\label{fig:tau}
\end{figure}

Given a path $p$, we say that a path $q$ is a \emph{right shortcut for $p$} if $q$ is a shortcut of $p$ and $q\subseteq\Right_p$. The following corollary can be proved by the same approach of Theorem~\ref{th_shortcut} and it is more useful for our purposes.

\begin{corollary}\label{cor:right_shortcut}
Let $\lambda$ be an $i$-path, for some $i\in[k]$. If there are no right shortcut of $\lambda$, then there exists no path $\lambda'\subseteq\Right_\lambda$ satisfying $w(\lambda')\leq w(\lambda)$.
\end{corollary}

\section{Finding an implicit representation of shortest paths}
\label{section:PIPPO_and_correctness}

We present here \PIPPO, that is a preprocessing step before computing distances between terminal pairs and listing non-crossing shortest paths. In Theorem~\ref{th:algoritmo(rightmost+redundant)} we prove the correctness of \PIPPO and explain some properties of found paths.

In Section~\ref{section:compute_distances_and_listing_paths}  we show that, by using the output of \PIPPO, the distances between terminal pairs can be computed in linear time and that the sequence of edges in a shortest path $\pi$ can be listed in $O(\max\{\ell,\ell\log(\frac{k}{\ell})\})$, where $\ell$ is the number of edges in $\pi$.
 

The main idea behind \PIPPO is the following. We build a set of shortest $i$-paths $\{\lambda_i\}_{i\in[k]}$, by finding $\lambda_i$ at iteration $i$, where the terminal pairs are numbered according to a postorder visit of $T_g$. In particular, at iteration $i$ we find the rightmost shortest $i$-path in $U_i=\bigcap_{j\in[i-1]} \Right_{\lambda_{j}}$, in the following way: first we set $\lambda_i$ as the leftmost $i$-path in $U_i$, then we update $\lambda_i$ by moving right through shortcuts (the order in which shortcuts are chosen is not relevant). When $\lambda_i$ has no more right shortcuts, then it is the rightmost shortest $i$-path in $U_i$ by Corollary~\ref{cor:right_shortcut}. 

Now we explain the implicit representation of path $\lambda_i$ found by \PIPPO, for an arbitrary $i\in[k]$. We define $C_i=\{j\in[k]\,|\,p(j)=i\}$ the set of children of $i$ in $T_g$. Let $X_i=\lambda_i\cap\bigcup_{j\in C_i}\lambda_j$, i.e., the graph formed by edges shared by $\lambda_i$ and its children. In Theorem~\ref{th:algoritmo(rightmost+redundant)}, we prove that $\{\lambda_i\}_{i\in[k]}$ is a set of single-touch paths and that $\lambda_i$ is a shortest $i$-path, for all $i\in[k]$. In our implicit representation, we give only edges of $\lambda_i\setminus X_i$, this representation has linear total size as proved in Theorem~\ref{th:algoritmo(rightmost+redundant)}.  Then we use it to compute distances of $\lambda_i$'s and to list every $\lambda_i$'s efficiently.

\begin{algorithm}[H]
\SetAlgorithmName{Algorithm \texttt{ImplicitPaths}}{}{}
\renewcommand{\thealgocf}{}
 \caption{}
 \KwIn{a plane graph $U$ and a set of well-formed terminal pairs $\{(s_i,t_i)\}_{i\in[k]}$ on the external face of $U$, where $U$ is the union of $k$ shortest paths, each joining $s_i$ and $t_i$, for $i\in[k]$, in a unknown undirected planar graph $G$ with positive edge weights}
 \KwOut{an implicit representation of a set of paths $\{\lambda_1,\ldots,\lambda_k\}$, where $\lambda_i$ is a shortest $s_it_i$ path, for $i\in[k]$}
{Compute $T_g$ and renumber the terminal pairs $(s_1,t_1),\ldots,(s_k,t_k)$ according to a postorder visit of $T_{g}$\label{line:T_g}\;
\For{$i=1,\ldots,k$}{
	Let $\lambda_{i}$ be the leftmost $i$-path in $U_i=\bigcap_{j\in[i-1]} \Right_{\lambda_{j}}$\label{line:lambda_{i,0}}\;
	\While{there exists a right shortcut  $\tau$ of $\lambda_{i}$ in $U_i$\label{line:shorcut_it_i}}{
	$\lambda_i:=\lambda_i\rtimes\tau$\label{line:shorcut_it_i_DOPOwhile}\;
	}
Compute $\lambda_i\setminus X_i$\;\label{line:X_i}
}	
   }
\end{algorithm}

To prove the correctness of \PIPPO in Theorem~\ref{th:algoritmo(rightmost+redundant)} we need the following preliminary lemma about properties of a set of pairwise non-crossing paths.

\begin{lemma}\label{lemma:a,b}
Let $\{\pi_i\}_{i\in[k]}$ be a set of pairwise non-crossing paths such that $\pi_i$ is an $i$-path, for $i\in[k]$. The following two statements hold:
\begin{enumerate}[label=\alph*)]\itemsep0em
\item\label{item:a} if $\ell\prec i\in[k]$, then $\pi_i\cap\pi_\ell\subseteq \bigcup_{j\in C_i}\pi_j$,
\item\label{item:b} if $i,j,\ell\in[k]$ are pairwise uncomparable, then $E(\pi_i\cap\pi_j\cap\pi_\ell)=\emptyset$.
\end{enumerate} 
\end{lemma}
\begin{proof}
If $\ell\in C_i$, then \ref{item:a} is trivial; otherwise it is implied by the non-crossing property. To prove~\ref{item:b}, let $c_i=\pi_r\circ\tau_r$, with $r\in\{i,j,\ell\}$, where $\tau_r$ is an auxiliary edge $s_r t_r$. By hypothesis, it holds that $c_i,c_j$ and $c_\ell$ are three closed simple curves for which no curve is contained in the interior region delimited by another curve. Thus the edge intersection is empty by an easy application of Jordan curve theorem~\cite{jordan}: if edge $e$ belongs to $c_i\cap c_j$, then $e$ is in the interior region bounded by the infinite face of $c_i\cup c_j$, thus $e$ does not belong to $c_\ell$. We note that the vertex intersection may be non-empty.\end{proof}

\begin{theorem}\label{th:algoritmo(rightmost+redundant)}
Let $\{\lambda_i\}_{i\in[k]}$ be the set of paths computed by \PIPPO. Then
\begin{enumerate}[label=\thetheorem.(\arabic*), ref=\thetheorem.(\arabic*),leftmargin=\widthof{10.(1)}+\labelsep]\itemsep0em
\item\label{item:lambda_is_rightmost} $\lambda_{i}$ is the rightmost shortest $i$-path in $U_i$, for $i\in[k]$,
\item\label{item:lambda_i_non_crossing_and_single_touch} $\{\lambda_i\}_{i\in[k]}$ is a set of single-touch paths,
\item\label{item:output_linear} $\sum_{i\in[k]}|E(\lambda_i\setminus X_i)|=O(|E(U)|)$. 
\end{enumerate}
\end{theorem}
\begin{proof}
We proceed by induction to prove the first statement. Trivially $\lambda_1$ is the rightmost shortest path in $U_1$ because of Corollary~\ref{cor:right_shortcut}. Let us assume that $\lambda_j$ is the rightmost shortest $j$-path in $U_j$, for $j\in[i-1]$, we have to prove that $\lambda_i$ is the rightmost shortest $i$-path in $U_i$.

In Line~\ref{line:lambda_{i,0}}, we initialize $\lambda_i$ as the leftmost $i$-path in $U_i$. By induction, at this step, there does not exist a path $p$ to the left of $\lambda_i$ shorter than $\lambda_i$. Otherwise $\lambda_i$ would cross a path $\lambda_j$, for some $j<i$, implying that $\lambda_j$ is not a shortest $j$-path. We conclude by the while cycle in Line~\ref{line:shorcut_it_i} and Corollary~\ref{cor:right_shortcut}.

Statement~\ref{item:lambda_i_non_crossing_and_single_touch} follows from~\ref{item:lambda_is_rightmost}; indeed, if $\lambda_i$ and $\lambda_j$ are not single-touch, for some $i,j\in[k]$, then~\ref{item:lambda_is_rightmost} is denied either for $\lambda_i$ or for $\lambda_j$.

We prove~\ref{item:output_linear} by showing that every edge $e$ is contained in at most two $\lambda_i$'s. Let us assume by contradiction that there exist $e\in E(G)$ and $i,j,\ell\in[k]$ satisfying $e\in(\lambda_i\setminus X_i)\cap(\lambda_j\setminus X_j)\cap(\lambda_\ell\setminus X_\ell)$. Then one among~\ref{item:a} and~\ref{item:b} in Lemma~\ref{lemma:a,b} is denied, absurdum.
\end{proof}

\section{Implementing \PIPPO in linear time}
\label{section:computational_complexity_PIPPO}

The main result of this section is explained in Theorem~\ref{th:PIPPO's_cost} which proves that \PIPPO can be executed in linear time. The section is split in four parts. In Subsection~\ref{sub:binarization} we binarize the genealogy tree in order to simplify the notation. In Subsection~\ref{sub:lambda} we show how to deal with $X_i$'s by using lists. In Subsection~\ref{sub:shortcut_linear_time} we prove that all shortcuts can be found in linear time and, finally, in Subsection~\ref{sub:computational_complexity} we prove Theorem~\ref{th:PIPPO's_cost}.


\subsection{Binarization of the Genealogy tree}\label{sub:binarization}

In order to simplify the treatment, we can assume that $T_g$ is binary in the following way. Given $i,j\in[k]$, we say that $\ell\lhd j$ if $s_1,s_j,t_j,s_\ell,t_\ell,t_1$ appear in this order on $f^\infty$. Given $i\in[k]$ satisfying $|C_i|=r\geq3$, we order its set of children $C_i=(c_1,\ldots,c_r)$, so that $c_j\lhd c_{j+1}$ for $j\in[r-1]$. If we add another couple $(s_{k+1},t_{k+1})$ of terminal pair so that $s_{k+1}=c_2$ and $t_{k+1}=t_r$, then $i$ has only two children $c_1$ and $k+1$. By repeating this argument, we obtain the binarization of $T_g$.

We must do some observations. Being $U$ connected, then there exists an $s_{k+1}t_{k+1}$ path that does not cross other paths. Moreover, this path is a shortest path in $U$ but it might be not a shortest path in $G$. We do not care about this because it is an auxiliary path. By repeating this reasoning to all $i\in[k]$ that have more than two children, we can assume that $T_g$ is binary. Finally, note that the number of terminal pairs becomes at most $2k$. In Figure~\ref{fig:gt_binary} the binarization of $T_g$ in Figure~\ref{fig:gt}.

\begin{figure}[h]
\captionsetup[subfigure]{justification=centering}
\centering
	\begin{subfigure}{4cm}
\begin{overpic}[width=6cm,percent]{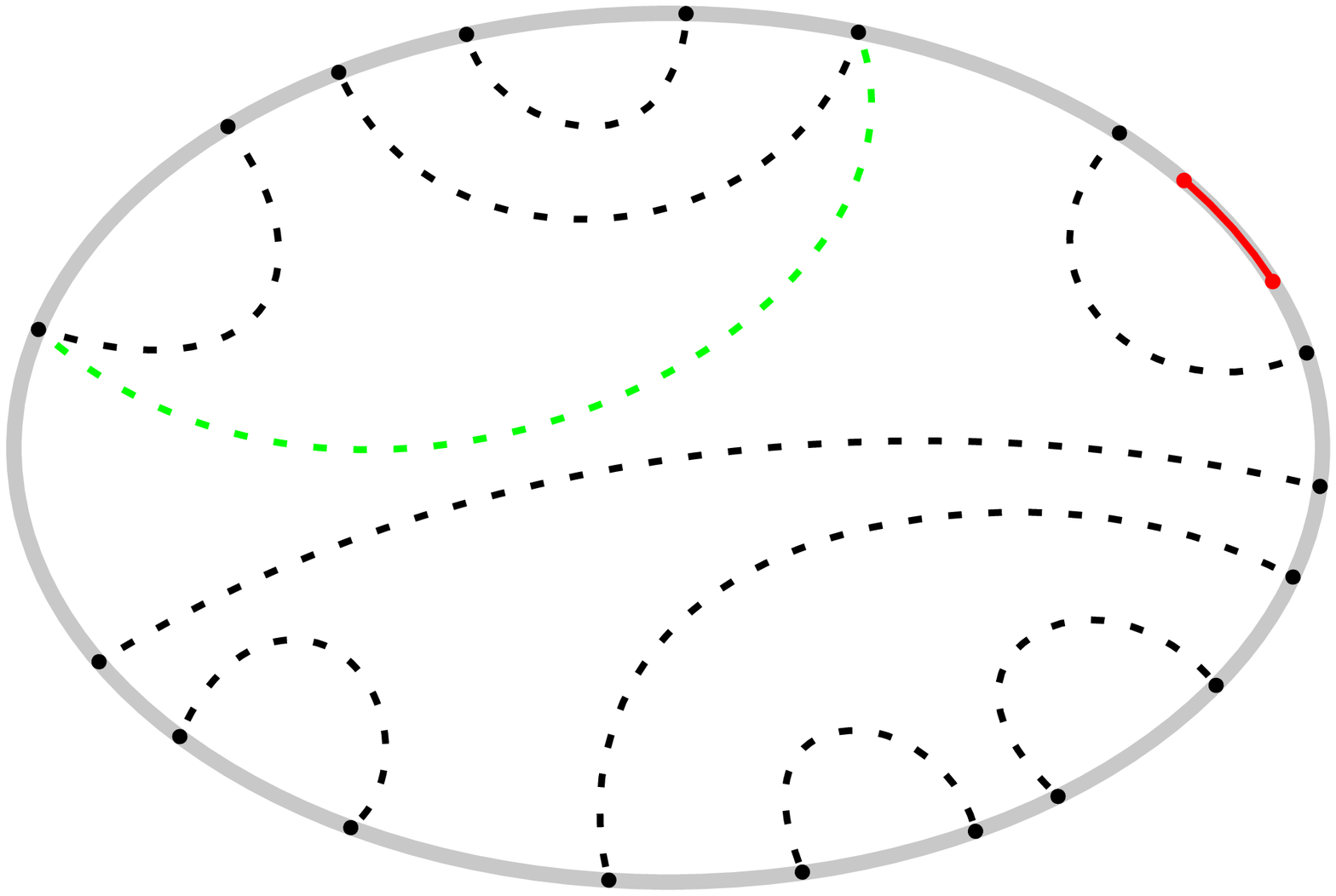}
\put(91,51){\color{red} $e^*$}
\put(-4.5,32){$f^\infty_G$}

\put(82,58){$t_1$}
\put(96,42){$s_1$}
\put(97.5,30.5){$s_2$}
\put(7,16.5){$t_2$}
\put(95.4,24){$s_3$}
\put(44.5,0){$t_3$}
\put(89.9,16.5){$s_4$}
\put(79,7){$t_4$}
\put(70.2,4.2){$s_5$}
\put(60,0.5){$t_5$}
\put(25,5){$s_6$}
\put(15,10){$t_6$}
\put(-13.5,43.5){$\color{green}{s_7}$ $=s_8$}
\put(15,59.5){$t_8$}
\put(24,63.5){$s_9$}
\put(64,65.6){$t_9=$ $\color{green}{t_7}$}
\put(34,66.5){$s_{10}$}
\put(50,67.5){$t_{10}$}

\end{overpic}
\end{subfigure}
\qquad\qquad\qquad\qquad\qquad\quad
	\begin{subfigure}{4cm}
\begin{overpic}[width=4cm,percent]{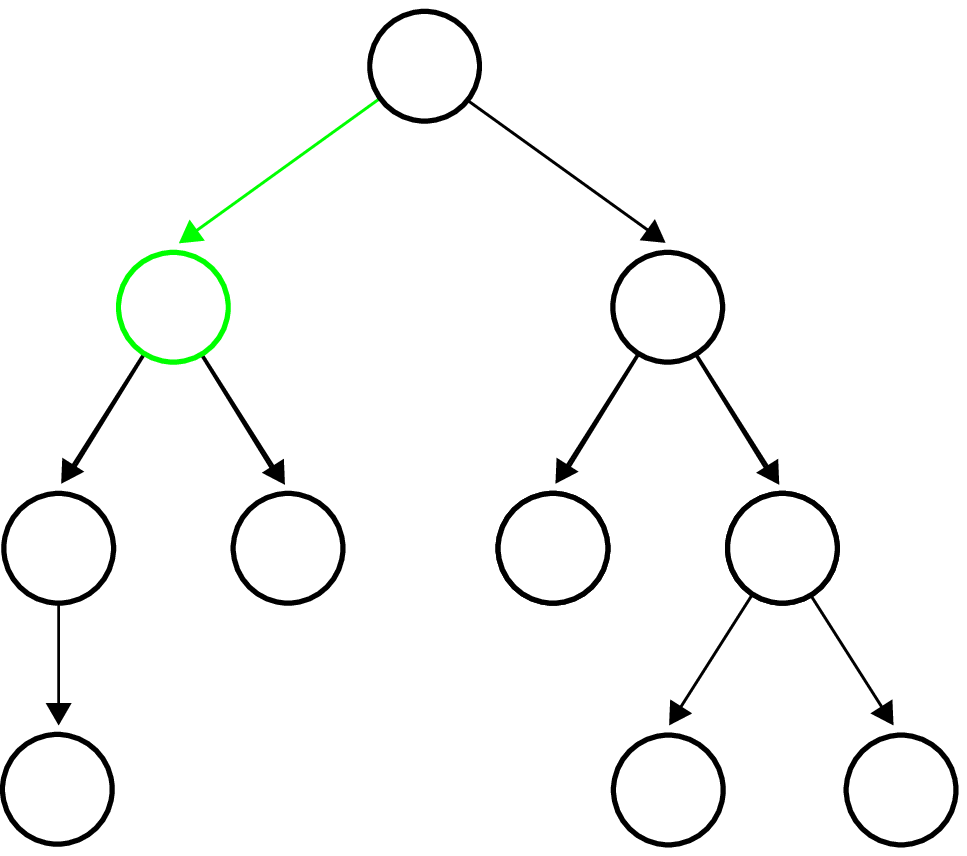}
\put(42.2,79.2){$1$}
\put(16.2,54){$7$}
\put(67.5,54){$2$}
\put(4.2,28.5){$9$}
\put(28.2,28.5){$8$}
\put(55.4,28.5){$6$}
\put(79.6,28.5){$3$}
\put(67.5,3.5){$5$}
\put(91.7,3.5){$4$}
\put(1.8,3.5){$10$}
\end{overpic}
\end{subfigure}
\caption{how to binarize the genealogy tree in Figure~\ref{fig:gt} by adding a couple of terminal pair (in green).}
\label{fig:gt_binary}
\end{figure}

\subsection{Dealing with $X_i$'s by using lists}\label{sub:lambda}

We note that $\sum_{i\in[k]}X_i$ might be $\Theta(|E(U)|^2)$. In this subsection we show how to deal with all $X_i$'s by using lists $L_{i,j}$'s (see Definition~\ref{def:baffetti}). This allows us to obtain $\lambda_i\setminus X_i$ at iteration $i$ in a time proportional to the number of new edges visited at iteration $i$ (see proof of Theorem~\ref{th:PIPPO's_cost}).

For $i\in[k]$, we define $\lambda_{i,0}$ as the initial $\lambda_i$ computed in Line~\ref{line:lambda_{i,0}}, and we define $\lambda_{i,j}$, as $\lambda_i$ after Line~\ref{line:shorcut_it_i_DOPOwhile} has been executed $j$ times,  i.e., after $j$ shortcuts have been considered. Moreover, we define $X_{i,j}=\lambda_{i,j}\cap\bigcup_{\ell\in C_i}\lambda_\ell$. Finally, let $M_i$ be the number of times that \PIPPO executes Line~\ref{line:shorcut_it_i_DOPOwhile} at iteration $i$, i.e., $\lambda_i=\lambda_{i,M_i}$.

For each $\lambda_{i,j}$, with $i\in[k]$ and $j\leq M_i$, we define now the sequence $L_{i,j}$ of edges that emanate right from vertices in $\lambda_{i,j}$, ordered as they are encountered walking on $\lambda_{i,j}$ from $s_{i}$ to $t_{i}$, see Definition~\ref{def:baffetti} and Figure~\ref{fig:baffetti}.

\begin{definition}
\label{def:baffetti}
Given $i\in[k]$ and $j\leq M_i$, we define $L_{i,j}=(e_1,e_2,\ldots,e_r)$ the list of edges in 
$R=\{xy\in \Right_{\lambda_{i,j}}\,|\,x\in V(\lambda_{i,j}) \text{ and } y\not\in V(\lambda_{i,j})\}$
%
ordered according to a rightmost visit from $s_i$ to $t_i$.
\end{definition}

In order to simplify the notation, for any $i\in[k]$, we denote $L_{i,M_i}$ by $L_i$. Figure~\ref{fig:baffetti} shows an example of Definition~\ref{def:baffetti}. Edges in $L_i$ are drawn in dotted lines.

\begin{figure}[h]
\centering
\begin{overpic}[width=4.5cm,percent]{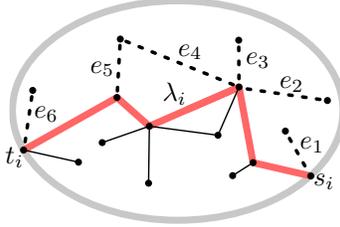}
     
\put(90,18){$s_i$}
\put(-1,24.5){$t_i$}
\put(45.5,43){$\lambda_i$}

\put(86,29){$e_1$}
\put(80,47){$e_2$}
\put(70,54){$e_3$}
\put(50,56.5){$e_4$}
\put(24,51.5){$e_5$}
\put(7.5,38){$e_6$}

\end{overpic}
\caption{an $i$-path  $\lambda_i$ in red and the set of edges $L_i=(e_1,e_2,e_3,e_4,e_5,e_6)$ in dotted lines. Edges incident on $\lambda$ in $\Left_{\lambda_i}$ are drawn in solid lines.}
\label{fig:baffetti}
\end{figure}

\begin{lemma}\label{lemma:left-left}
All edges in $\Left_{\lambda_{i,0}}\cap \bigcap_{j\in C_i}\Right_{j}$ are in $\lambda_{i,0}\cup \bigcup_{j\in C_i} \lambda_j$.
\end{lemma}
\begin{proof}
Let us assume by contradiction that there exists $e\in\Left_{\lambda_{i,0}}\cap \bigcap_{j\in C_i}\Right_{j}$ satisfying $e\not\in\lambda_{i,0}\cup \bigcup_{j\in C_i} \lambda_j$. Let $P=\{\sigma_1,\ldots,\sigma_k\}\in\mathcal{P}$ be a complete path-set and let us assume that $\sigma_j$ is a shortest $j$-path, for all $j\in[k]$. Let $\ell\in[k]$ satisfy $e\in\sigma_\ell$.

Now, if $\ell$ is a descendant of $i$, then $\ell$ is a descendant of $j$, for some $j\in C_i$ (possibly, $\ell=j$). Thus $\lambda_j RE \sigma_\ell$ is a shortest $j$-path right of $\lambda_j$, absurdum because of~\ref{item:lambda_is_rightmost}. Otherwise, $\ell$ is an ancestor of $i$ or they are uncomparable, in both cases $\lambda_{i,0} LE \sigma_\ell$ is a shortest $i$-path left of $\lambda_{i,0}$, absurdum because of definition of $\lambda_{i,0}$.
\end{proof}

The single-touch property implies that the intersection among $L_{i,j}$ and $L_\ell$, for $\ell\in C_i$ and $j\leq M_i$, is always a sublist of $L_\ell$. We recall that $i$ has at most two children in $T_g$. Thanks to Lemma~\ref{lemma:left-left} and connection of graph $U$, there are few cases and they are all distinguishable by looking to $L_j$'s for $j\in C_i$ as explained in the following remark.

\begin{remark}\label{remark:lists}
The following statements hold for all $i\in[k]$:
\begin{itemize}\itemsep0em
\item if $i$ has only one child $j$, then $|L_j|\geq2$ and $\uuu{\lambda_{i,0}}$ passes through the first and last edge of $L_j$, see Figure~\ref{fig:lambda}.(\subref{fig:lambda_1})
\item if $i$ has two children $j$ and $\ell$ without vertex intersection, then $|L_j|\geq2$, $|L_\ell|\geq2$ and $\uuu{\lambda_{i,0}}$ passes through the first and last edge of $L_j$ and $L_\ell$, see Figure~\ref{fig:lambda}.(\subref{fig:lambda_2})
\item if $i$ has two children $j$ and $\ell$, $\ell\lhd j$, with non-empty vertex intersection then $|L_\ell|\geq2$ and $|L_j|\geq2$. There are three cases
\begin{itemize}\itemsep0em
\item if $|L_j|>2$ and $|L_\ell|>2$, then $\uuu{\lambda_{i,0}}$ passes through the first and penultimate edge of $L_j$ and through the second and last edge of $L_\ell$, see Figure~\ref{fig:lambda}.(\subref{fig:lambda_3})
\item if $|L_j|>2$ and $|L_\ell|=2$, then $\uuu{\lambda_{i,0}}$ passes through the first and third last edge of $L_j$, see Figure~\ref{fig:lambda}.(\subref{fig:lambda_4})
\item if $|L_j|=2$ and $|L_\ell|>2$, then $\uuu{\lambda_{i,0}}$ passes through the third and last edge of $L_\ell$, it is symmetric to Figure~\ref{fig:lambda}.(\subref{fig:lambda_4}).
\end{itemize}
\end{itemize}
\end{remark}

\begin{figure}[h]
\captionsetup[subfigure]{justification=centering}
\centering
	\begin{subfigure}{4.5cm}
\begin{overpic}[width=4.5cm,percent]{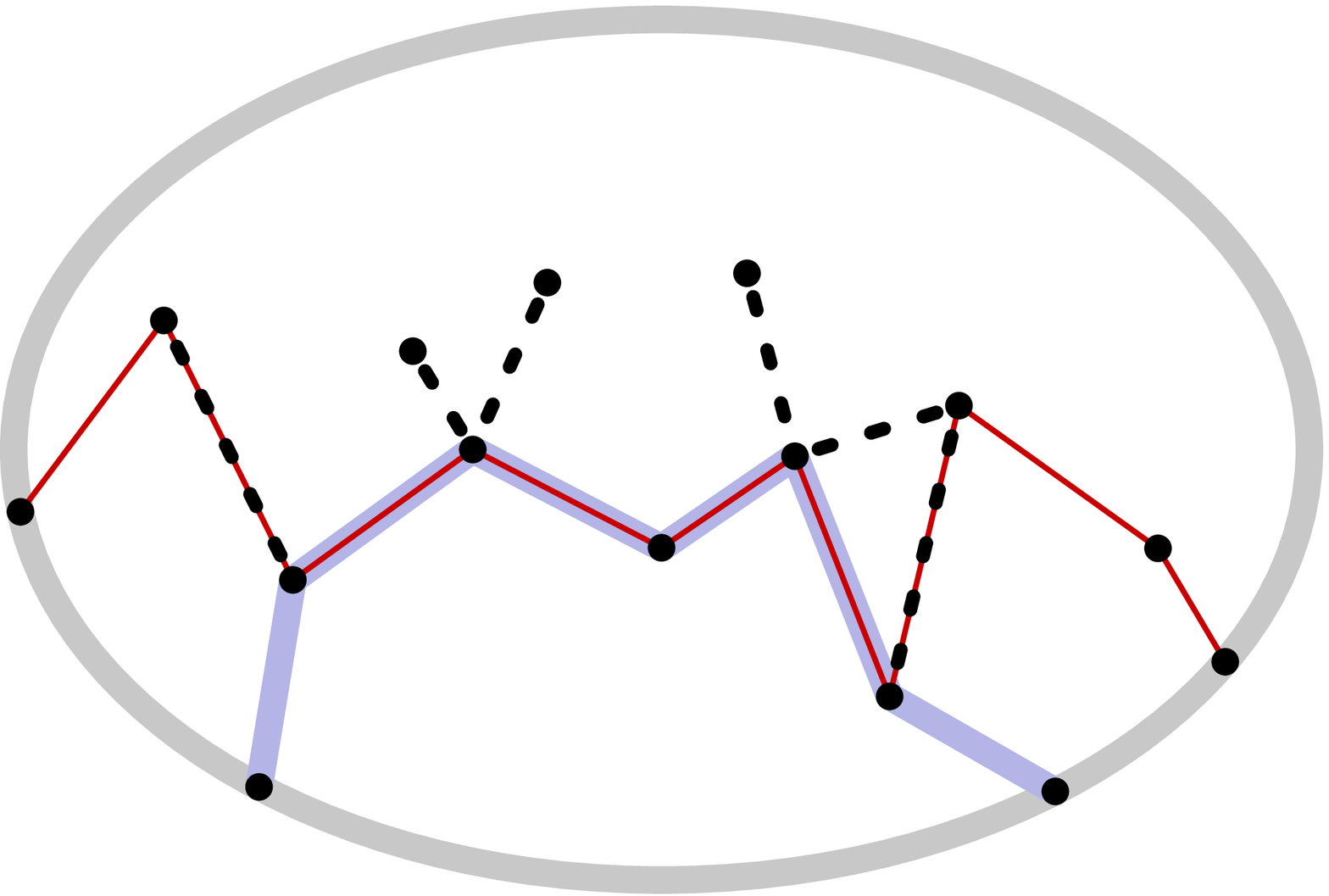}
\put(93,12){$s_i$}
\put(-5,26){$t_i$}
\put(80,3){$s_j$}
\put(17,1){$t_j$}
\end{overpic}
\caption{}\label{fig:lambda_1}
\end{subfigure}
\qquad
	\begin{subfigure}{4.5cm}
\begin{overpic}[width=4.5cm,percent]{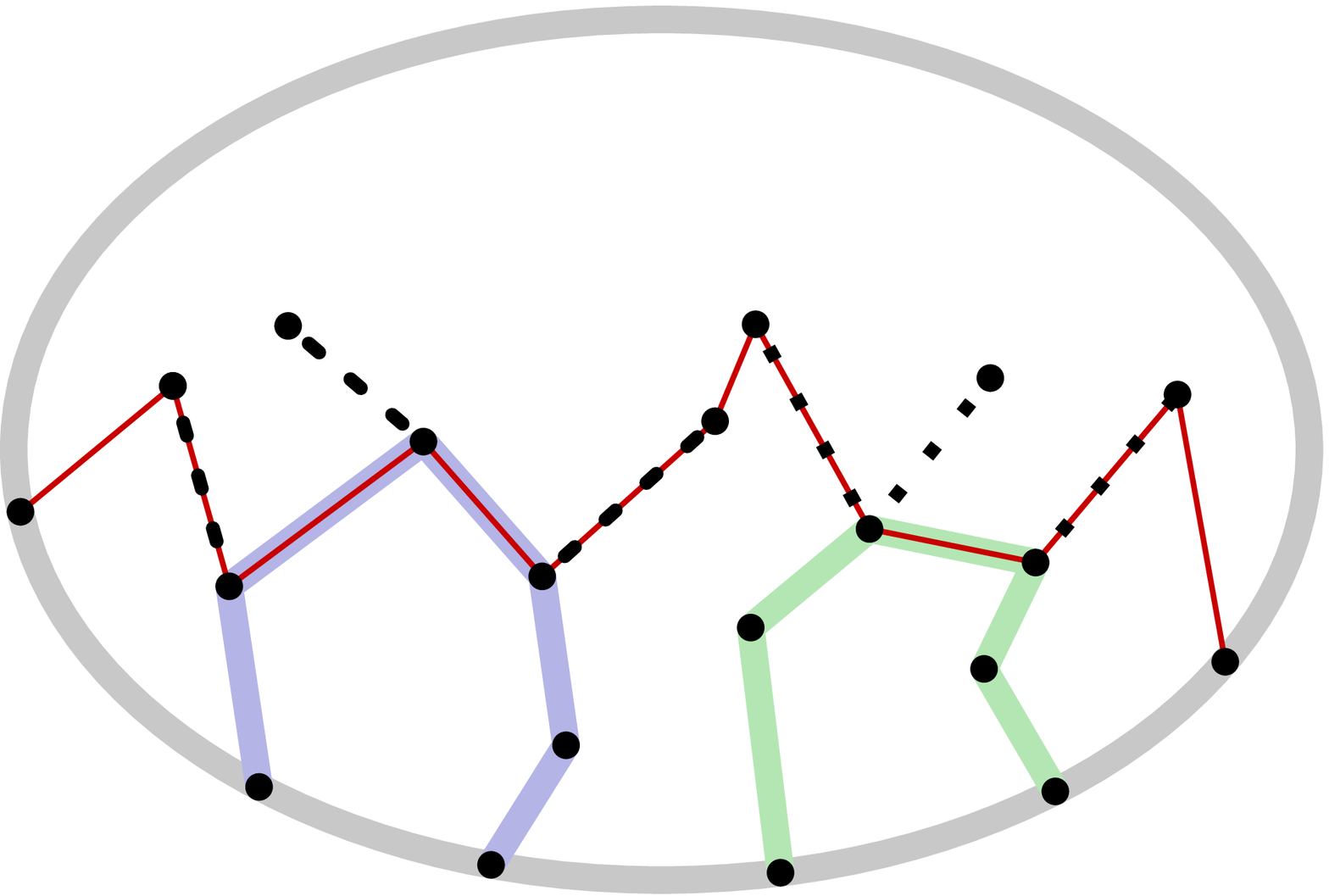}

\put(93,12){$s_i$}
\put(-5,26){$t_i$}
\put(80,3){$s_j$}
\put(58,-5){$t_j$}
\put(34.5,-3.5){$s_\ell$}
\put(17,1){$t_\ell$}

\end{overpic}
\caption{}\label{fig:lambda_2}
\end{subfigure}
\qquad\qquad
	\begin{subfigure}{4.5cm}
\begin{overpic}[width=4.5cm,percent]{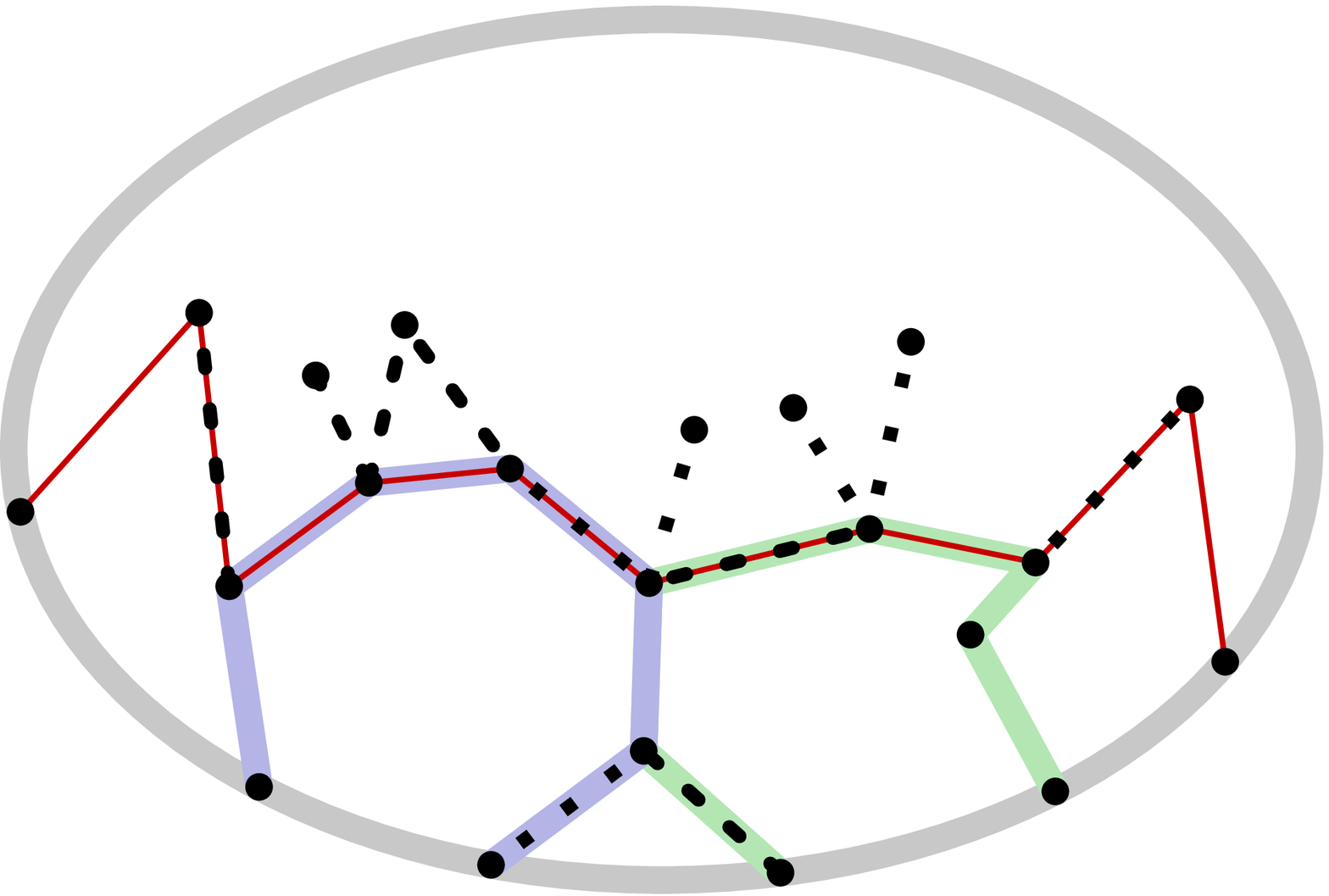}
\put(93,12){$s_i$}
\put(-5,26){$t_i$}
\put(80,3){$s_j$}
\put(58,-5){$t_j$}
\put(34.5,-3.5){$s_\ell$}
\put(17,1){$t_\ell$}
\end{overpic}
\caption{}\label{fig:lambda_3}
\end{subfigure}
\qquad
	\begin{subfigure}{4.5cm}
\begin{overpic}[width=4.5cm,percent]{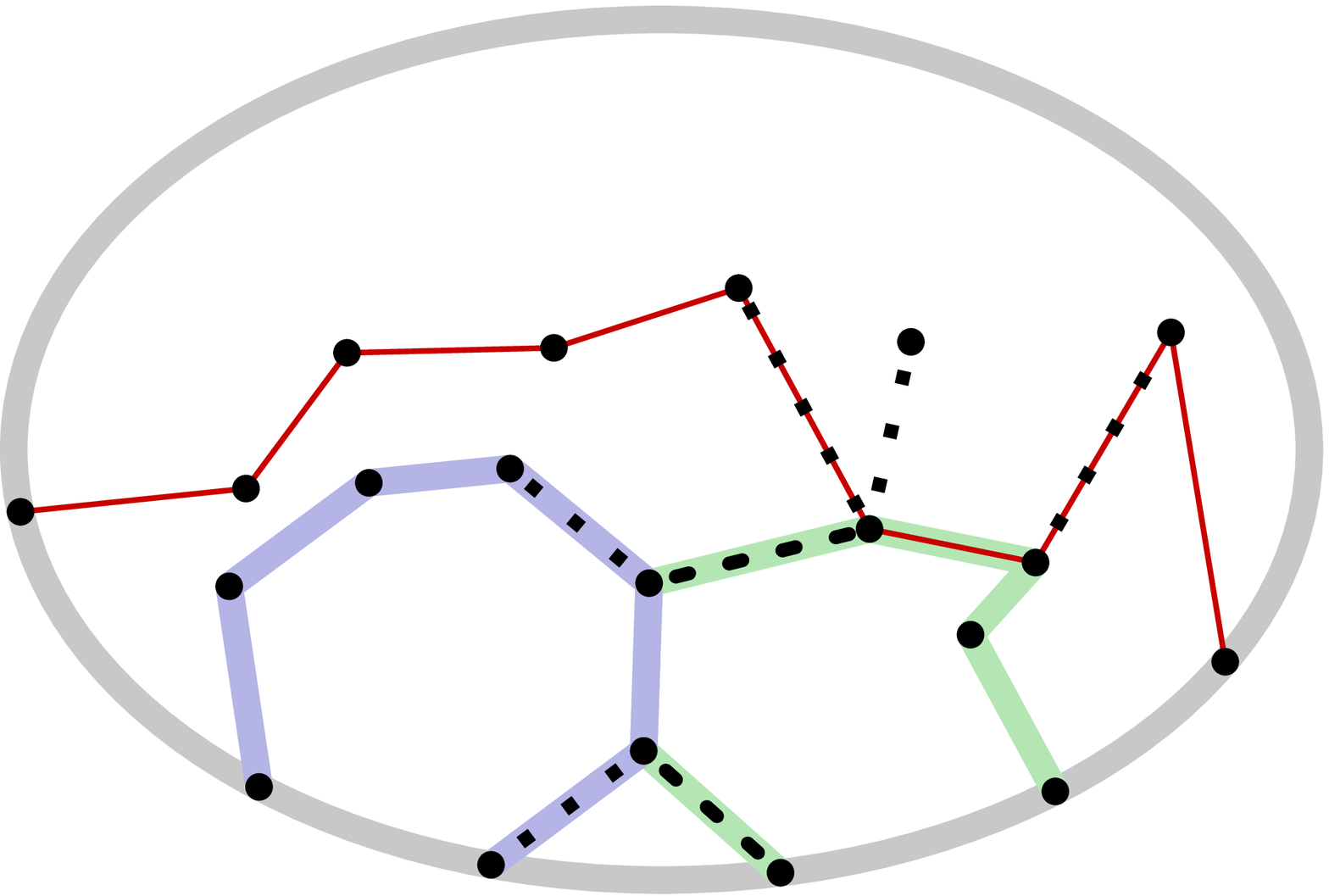}
\put(93,12){$s_i$}
\put(-5,26){$t_i$}
\put(80,3){$s_j$}
\put(58,-5){$t_j$}
\put(34.5,-3.5){$s_\ell$}
\put(17,1){$t_\ell$}
\end{overpic}
\caption{}\label{fig:lambda_4}
\end{subfigure}
  \caption{the four cases described in Remark~\ref{remark:lists}, $\lambda_{i,0}$ is drawn in red.}\label{fig:lambda}
\end{figure}

\subsection{Computing all shortcuts}
\label{sub:shortcut_linear_time}

In this subsection we show how to compute all shortcuts of \PIPPO in linear time (see Lemma~\ref{lemma:shortcuts_linear}). The following lemma is a preliminary result and it shows that in a particular case we know the existence of a right shortcut.

\begin{lemma}\label{lemma:shortcuts}
Let $i\in[k]$, $j<M_i$ and $f\in\mathcal{F}$. If $\lambda_{i,j}\cap \partial f$ is not a path and $\partial f\subseteq\Right_{\lambda_{i,j}}$, then there exists a right shortcut for $\lambda_{i,j}$ that is a subpath of $\partial f$.
\end{lemma}
\begin{proof}
By hypothesis there exist two edges $e=uv,e'=xy$ of $\partial f$ such that $e,e'\in\lambda_{i,j}$, $u,v,x,y$ appear in this order going from in $\lambda_{i,j}$, and $\lambda_{i,j}[v,x]$ does not intersect on edges $\partial f$. Let $\sigma$ be the subpath of $\partial f$ such that $v,x\in V(\sigma)$ and $u,y\not\in V(\sigma)$. We prove that $\sigma$ is a right shortcut of $\lambda_{i,j}$.

First we observe that $\sigma\subseteq\Right_{\lambda_{i,j}}$ because $\partial f\subseteq\Right_{\lambda_{i,j}}$. Let $P\in\mathcal{P}$ be a complete path-set. Let $\{q_1,q_2,\ldots,q_r\}\subseteq P$ be the minimal subset of paths of $P$ such that $\sigma\subseteq\bigcup_{\ell\in[r]}q_\ell$. Let $\lambda_1=\uuu{\lambda_{i,j}} RE q_1$, $\lambda_2=\uuu{\lambda_1} RE q_2$, $\lambda_3=\uuu{\lambda_2} RE q_3$, $\ldots$, $\lambda_r=\uuu{\lambda_{r-1}} RE q_r$. It is clear that $\sigma\subseteq\lambda_r$, moreover $w(\lambda_r)\leq w(\uuu{\lambda_{i,j}})$ because $\lambda_r$ is obtained from $\uuu{\lambda_{i,j}}$ by right envelopes with shortest paths. Thus $\sigma$ is a shortcut of $\uuu{\lambda_{i,j}}$ and the proof is completed.
\end{proof}

\begin{lemma}\label{lemma:shortcuts_linear}
We can compute all shortcuts of \PIPPO in linear time.
\end{lemma}
\begin{proof}
Let $f\in\mathcal{F}$. If $\uuu{\lambda_{i,0}}\cap f\subseteq\uuu{\lambda_j}$, for some $i\in[k]$ and $j\in C_i$, then there is not any right shortcut for $\uuu{\lambda_{i,j}}$ in $f$, otherwise~\ref{item:lambda_is_rightmost} would be denied. Thus in the whole \PIPPO, we ask if there is a right shortcut in $f$ every time a new edge of $f$ is visited. To complete the proof, we have to answer in $O(1)$ every time.

We preprocess $U$ so that for every face $f\in\mathcal{F}$ and for every vertices $u,v \in\partial f$ we can compute the weight of the clockwise path on $\partial f$ from $u$ to $v$ and the weight of the counterclockwise path on $\partial f$ from $u$ to $v$ in $O(1)$. It is sufficient to compute $w(\partial f)$ and the weight of the clockwise path on $\partial f$ from a fixed vertex $u$ to any vertex $v$ in $\partial f$. This preprocessing clearly costs $O(|V(\partial f)|)$ time, for every $f\in\mathcal{F}$, thus it costs linear time for the whole set of faces $\mathcal{F}$.

Now we observe that, by Lemma~\ref{lemma:shortcuts}, $\uuu{\lambda_{i,j}}\cap\partial f$ can be considered always connected when $i$ and $j$ grow; if it is not connected, then we overcome this problem in $O(1)$ by Lemma~\ref{lemma:shortcuts}. Hence we can decide whether $w(\uuu{\lambda_{i,j}}\cap\partial f)\geq w(\partial f\setminus(\uuu{\lambda_{i,j}}\cap \partial f))$ in $O(1)$ time using the above preprocessing.
\end{proof}


\subsection{Overall complexity}\label{sub:computational_complexity}

Now we have all necessary results to prove the main result of this section.

\begin{theorem}\label{th:PIPPO's_cost}
\PIPPOMaiuscolo requires linear time.
\end{theorem}
\begin{proof}
We observe that computing Line~\ref{line:T_g} costs total linear time. Let $i\in[k]$. Finding $\lambda_{i,0}$ costs $O(|\lambda_{i,0}\setminus X_{i,0}|)$ by Remark~\ref{remark:lists}, computing $L_{i,0}$ costs $O(|L_{i,0}\setminus \bigcup_{j\in C_i}L_i|)+O(1)$ because we are working with lists and we join them in $O(1)$ by Remark~\ref{remark:lists}. We obtain $\lambda_i$ from $\lambda_{i,0}$ in $O(|\lambda_i\setminus\lambda_{i,0}|)$ plus the cost of finding all right shortcuts at iteration $i$. Similarly, we obtain $L_i$ from $L_{i,0}$ in $O(|L_i\setminus L_{i,0}|)$ plus the cost of finding all right shortcuts at iteration $i$. By generality of $i$, we have proved that we can find what we need (except for the shortcuts) by visiting every edge a constant number of times. We conclude the proof by Lemma~\ref{lemma:shortcuts_linear}.
\end{proof}

\section{Computing distances and listing paths}
\label{section:compute_distances_and_listing_paths}

By using the implicit representation of paths given by \PIPPO, we compute distances between terminal pairs in linear total time, as showed in the following theorem.

\begin{theorem}
\label{th:distances_in_O(n)}
We compute $w(\lambda_i)$, for all $i\in[k]$, in linear total time.
\end{theorem}
\begin{proof}
For convenience, given $i\in[k]$ and $j\in C_i$,  we say that $\lambda_j$ is a \emph{child of $\lambda_i$}. By the  binarization of $T_g$ made in Section~\ref{section:computational_complexity_PIPPO} there are three cases: $\lambda_i$ does not intersect on vertices its children, $\lambda_i$ intersects exactly one child, $\lambda_i$ intersects two children. It's clear that if the first case happens, than we compute $w(\lambda_i)$ without visiting edges belonging to its children. Now we explain the second case, the third is similar.

Let us assume that $\lambda_i=p^i\circ\lambda_z[c_z,d_z]\circ q^i$ for some $z\in C_i$ and some (possibly non-empty) paths $p^i$ and $q^i$ that not intersect the children of $\lambda_i$. To compute $w(\lambda_i)$, we first visit $p^i$ and $q^i$ in order to obtain $w(p^i)$ and $w(q^i)$. Then we observe that $w(\lambda_z[c_z,d_z])=w(\lambda_z)-\big(w(\lambda_z[s_z,c_z])+w(\lambda_z[d_z,t_z])\big)$, thus it suffices to visit $\lambda_z[s_z,c_z]$ and $\lambda_z[d_z,t_z]$. Note that if we have already computed $w(\lambda_z[s_z,x])$ (resp., $w(\lambda_z[y,t_z])$)  for some vertex $x\in \lambda_z[s_z,c_z]$ (resp., $y\in \lambda_z[d_z,t_z]$), then we visit only edges in $\lambda_z[x,c_z]$ (resp., $\lambda_z[y,t_z]$).

In this way and edge in $\lambda_i$ is visited at most one time at iteration $i$ and one time in all other iterations $j$, for $i\prec j$. By~\ref{item:b} in Lemma~\ref{lemma:a,b} and generality of $i$, it follows that every edge is visited at most four times.
\end{proof}

Now we study the problem of listing the edges in an $i$-path, for some $i\in[k]$, after the execution of \PIPPO. We want to underline the importance of the single-touch property. In Figure~\ref{fig:zigzag}, in (\subref{fig:zigzag1})  four shortest paths are drawn (the graph is unit-weighted), we observe that the single-touch property is clearly not satisfied. A single-touch version of the previous four paths is drawn in (\subref{fig:zigzag2}); it can be obtained by \PIPPO choosing $i^*=1$ in the genealogy tree. It is clear that the problem of listing the edges in a path in this second case is easier. We stress that in general cases, the union of a set of single-touch paths can form cycles, see Figure~\ref{fig:errore_giappo} for an example.

\begin{figure}[h]
\captionsetup[subfigure]{justification=centering}
\centering
	\begin{subfigure}{5cm}
\begin{overpic}[width=5cm,percent]{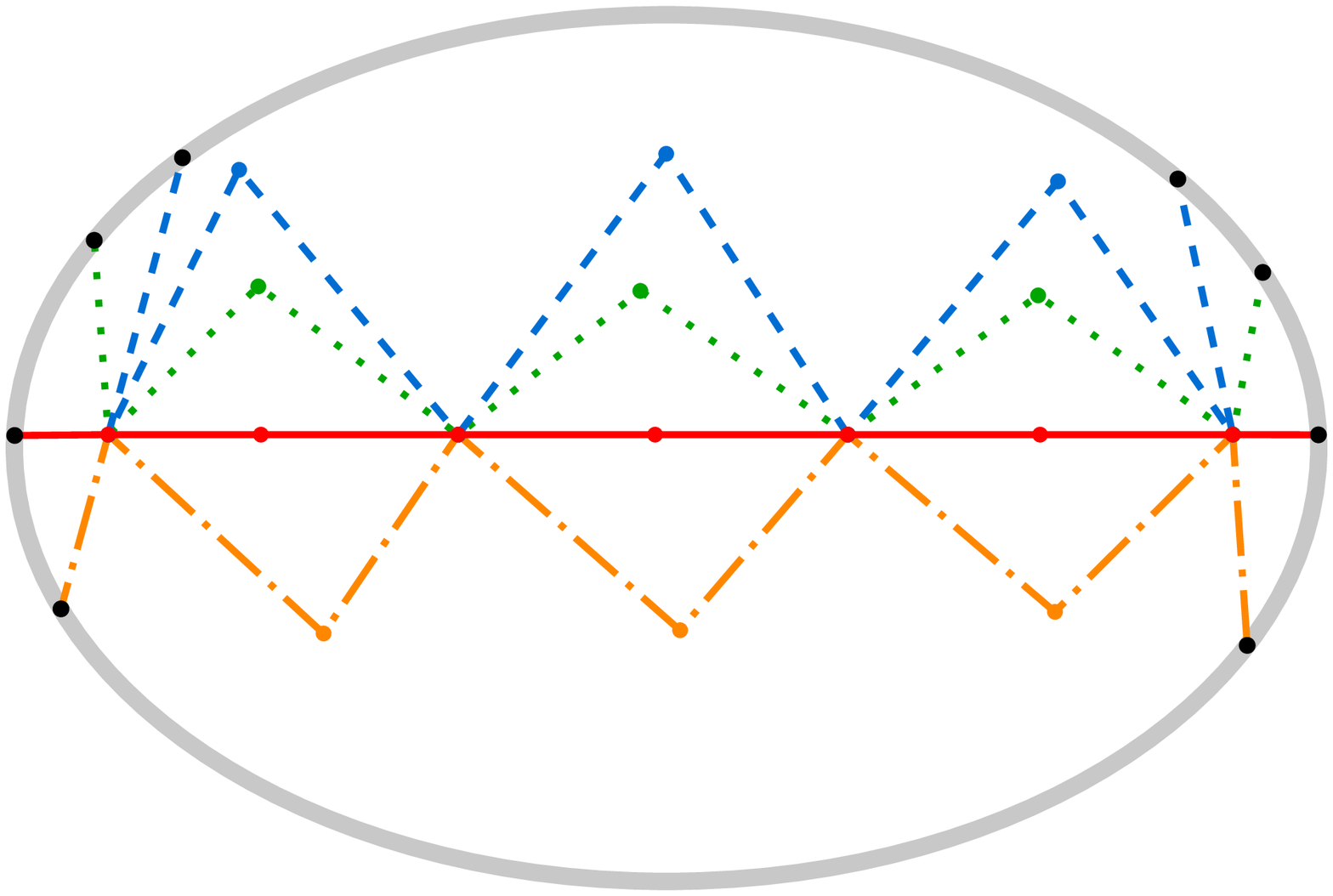}
\put(89.5,55){$s_1$}
\put(7.8,57.2){$t_1$}
\put(96,48.2){$s_2$}
\put(1.3,51.2){$t_2$}
\put(101,34.5){$s_3$}
\put(-6,34.5){$t_3$}
\put(94,15.3){$s_4$}
\put(0.3,16.5){$t_4$}
\end{overpic}
\caption{}\label{fig:zigzag1}
\end{subfigure}
\qquad\qquad
	\begin{subfigure}{5cm}
\begin{overpic}[width=5cm,percent]{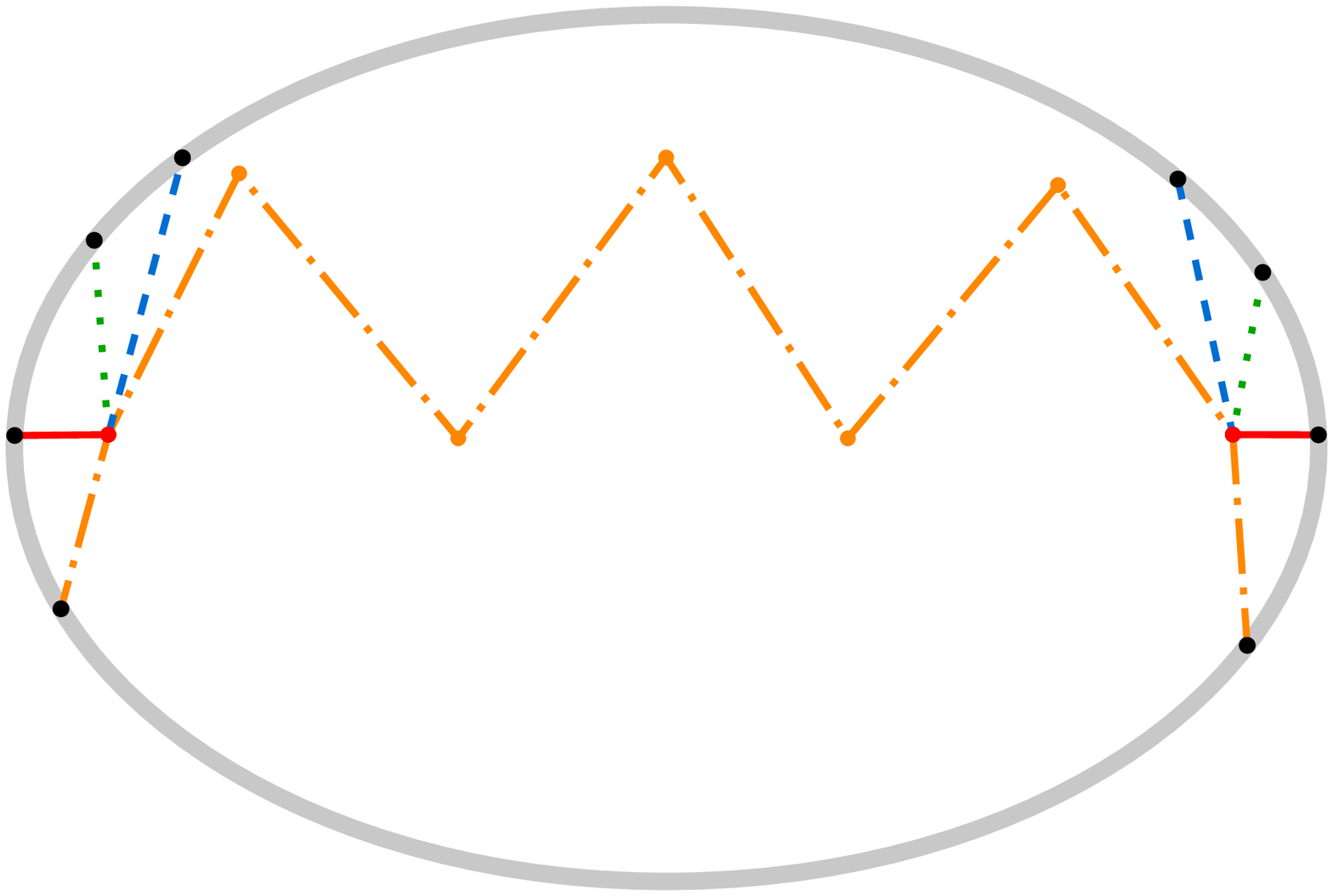}
\put(89.5,55){$s_1$}
\put(7.8,57.2){$t_1$}
\put(96,48.2){$s_2$}
\put(1.3,51.2){$t_2$}
\put(101,34.5){$s_3$}
\put(-6,34.5){$t_3$}
\put(94,15.3){$s_4$}
\put(0.3,16.5){$t_4$}
\end{overpic}
\caption{}\label{fig:zigzag2}
\end{subfigure}
\caption{(\subref{fig:zigzag1}) the union of shortest $i$-paths, for $i\in[4]$, in unit-weighted graph, every different path has different style, (\subref{fig:zigzag2}) the union of $\{\lambda_i\}_{i\in[4]}$, the output paths of \PIPPO, by assuming that $1$ is the root of the genealogy tree.}
\label{fig:zigzag}
\end{figure}

\begin{theorem}\label{th:cost_find_pi_i_from_output}
After $O(n)$ time preprocessing, each shortest path $\lambda_i$, for $i\in[k]$,  can be listed in $O(\max\{\ell_i,$ $\ell_i\log(\frac{k}{\ell_i})\})$ time, where $\ell_i$ is the number of edges of $\lambda_i$.
\end{theorem}
\begin{proof}
For any $i\in[k]$, we denote by $\dd{\lambda_i}$ the oriented version of $\lambda_i$ from $s_i$ to $t_i$. We denote by $\dd{uv}$ the dart from $u$ to $v$ and for every dart $\dd{uv}$ we define $\head(\dd{uv})=v$.

By a small change in \PIPPO in Line~\ref{line:X_i}, we introduce a function \texttt{Mark} that marks a dart $d$ with $i$ if and only if the $d$ is used for the first time in the execution of \PIPPO at iteration $i$. It means that \texttt{Mark}$(d)=i$ if and only if $d$ belongs to $\dd{\lambda_i}$ and $d$ does not belong to $\dd{\lambda_j}$, for all $j\prec i$; this follows from an oriented version of Lemma~\ref{lemma:a,b}, indeed, if $i$ and $j$ are uncomparable, then $\dd{\lambda_i}$ and $\dd{\lambda_j}$ share no darts. This function can be executed within the same time bound of \PIPPO.
Now we explain how to find darts in $\dd{\lambda_i}$.

Let us assume that $(d_1,\ldots,d_{\ell_i})$ are the ordered darts in $\dd{\lambda_i}$. Let $v=\head(d_{j-1})$, and let us assume that $deg(v)=r$ in the graph $\bigcup_{j\in[k]}\lambda_j$. We claim that if we know $d_{j-1}$, then we find $d_{j}$ in $O(\log r)$. First we order the outgoing darts in $v$ in clockwise order starting in $d_{j-1}$, thus let \emph{Out}$_v=(g_1,\ldots,g_r)$ be this ordered set (this order is given by the embedding of the input plane graph). We observe that all darts in \emph{Out}$_v$ that are in $\Left_{\lambda_i}$ are in $\dd{\lambda_w}$ for some $w\leq i$, thus $\texttt{Mark}(d)\leq i$ for all $d\in$ \emph{Out}$_v$ satisfying $d\in\Left_{\lambda_i}$. Similarly, all darts in \emph{Out}$_v$ that are $\Right_{\lambda_i}$ are in $\dd{\lambda_z}$ for some $z\geq i$, thus $\texttt{Mark}(d)\geq i$ for all $d\in$ \emph{Out}$_v$ satisfying $d\in\Right_{\lambda_i}$. Using this observation, we have to find the unique $l\in[r]$ such that $\texttt{Mark}(g_l)\leq i$ and $\texttt{Mark}(g_{l+1})>i$. This can be done in $O(\log r)$ by using a binary search.

Being the $\dd{\lambda_i}$'s pairwise single-touch, then $\sum_{v\in V(\lambda_i)} deg(v)\leq 2k$, where the equality holds if and only if every $\dd{\lambda_j}$, for $j\neq i$, intersects on vertices $\dd{\lambda_i}$ exactly two times, that is the maximum allowed by the single-touch property.

Finally, if $2k\leq\ell_i$, then we list $\dd{\lambda_i}$ in $O(\ell_i)$ because the  binary searches of the correct darts do not require more than $O(k)$ time, otherwise  we note that
$$\sum_{\substack{j=1,\ldots,\ell\\ a_1+\ldots+a_\ell\leq2k}}\log a_j\leq \ell\log\left(\frac{2k}{\ell}\right),$$
so the time complexity follows.
\end{proof}

\section{Conclusions}
\label{section:conclusions}

In this paper we extend the result of Takahashi \emph{et al.}~\cite{giappo2} by computing the lengths of non-crossing shortest paths in  undirected plane graphs also in the general case when the union of shortest paths is not a forest. Moreover, we provide an algorithm for listing the sequence of edges of each path in $O(\max\{\ell,\ell\log(\frac{k}{\ell})\})$, where $\ell$ is the number of edges in the shortest path. Our results are based on the local concept of shortcut.

All results of this paper can be easily applied in a geometric setting, where it is asked to search for paths in polygons instead of plane graphs. The same results can be extended to the case of terminal pairs lying on two distinct faces, by the same argument shown in~\cite{giappo2}.

We left open the problem of listing a shortest path in $O(\ell)$ time and finding the union of non-crossing shortest paths in $o(n\log k)$ time.

\bibliographystyle{siam}
\bibliography{biblio.bib}

\end{document}